\documentclass{easychair}
\usepackage{hyperref}

\usepackage{amsmath}
\usepackage{amssymb}
\usepackage{amsthm}
\usepackage{tikz}
\usepackage{pgfplots}
\pgfplotsset{compat=1.18}
\usepackage{booktabs}

\newcommand{\pand}{\mathbin{\land^\textsf{p}}}
\newcommand{\por}{\mathbin{\lor^\textsf{p}}}
\DeclareMathOperator*{\Pand}{\bigwedge^\textsf{p}}
\DeclareMathOperator*{\Por}{\bigvee^\textsf{p}}
\newcommand{\boolnot}{\neg}

\newcommand{\obar}[1]{\overline{#1}}

\newcommand{\lit}{\ell}

\newcommand{\varset}{X}
\newcommand{\exvarset}{Z}
\newcommand{\dependencyset}{{\cal D}}

\newcommand{\ring}{{\cal R}}
\newcommand{\semiring}{{\cal S}}
\newcommand{\dset}{{\cal A}}
\newcommand{\rep}{\textbf{R}}
\newcommand{\srep}{\textbf{S}}
\newcommand{\radd}{+}
\newcommand{\rmul}{\times}
\newcommand{\addident}{\textbf{0}}
\newcommand{\mulident}{\textbf{1}}
\newcommand{\imply}{\Rightarrow}
\newcommand{\ifandonlyif}{\Leftrightarrow}
\newcommand{\eequiv}{\ifandonlyif_{\varset}}
\newcommand{\drational}{\textbf{Q}_{2,5}}

\newcommand{\extend}[1]{#1^{*}}
\newcommand{\assign}{\alpha}
\newcommand{\eassign}{\extend{\alpha}}
\newcommand{\passign}{\rho}
\newcommand{\lassign}{\beta}
\newcommand{\uassign}{{\cal U}}
\newcommand{\modelset}{{\cal M}}

\newcommand{\indegree}{\textrm{indegree}}
\newcommand{\outdegree}{\textrm{outdegree}}
\newcommand{\validate}{\textsf{validate}}
\newcommand{\prov}{\textrm{Prov}}
\newcommand{\inputformula}{\phi_I}
\newcommand{\pogformula}{\theta_P}
\newcommand{\treesize}{T}

\newcommand{\makenode}[1]{\mathbf{#1}}
\newcommand{\nodeu}{\makenode{u}}
\newcommand{\nodev}{\makenode{v}}

\newcommand{\nodes}{\makenode{s}}
\newcommand{\nodep}{\makenode{p}}
\newcommand{\noder}{\makenode{r}}

\newcommand{\simplify}[2]{#1|_{#2}}

\newcommand{\progname}[1]{\textsc{#1}}
\newcommand{\dfour}{\progname{D4}}
\newcommand{\cdfour}{\progname{CD4}}

\newcommand{\cadical}{\progname{CaDiCal}}
\newcommand{\dtrim}{\progname{drat-trim}}
\newcommand{\ltrim}{\progname{lrat-trim}}

\newcommand{\mice}{MICE}
\newcommand{\Mice}{\mice}

\newcommand{\decdnnf}{decision-DNNF}
\newcommand{\detdnnf}{d-DNNF}

\definecolor{redorange}{rgb}{0.878431, 0.235294, 0.192157}
\definecolor{lightblue}{rgb}{0.552941, 0.72549, 0.792157}
\definecolor{clearyellow}{rgb}{0.964706, 0.745098, 0}
\definecolor{clearorange}{rgb}{0.917647, 0.462745, 0}
\definecolor{mildgray}{rgb}{0.54902, 0.509804, 0.47451}
\definecolor{softblue}{rgb}{0.643137, 0.858824, 0.909804}
\definecolor{bluegray}{rgb}{0.141176, 0.313725, 0.603922}
\definecolor{lightgreen}{rgb}{0.709804, 0.741176, 0}
\definecolor{darkgreen}{rgb}{0.152941, 0.576471, 0.172549}
\definecolor{redpurple}{rgb}{0.835294, 0, 0.196078}
\definecolor{midblue}{rgb}{0, 0.592157, 0.662745}
\definecolor{clearpurple}{rgb}{0.67451, 0.0784314, 0.352941}
\definecolor{browngreen}{rgb}{0.333333, 0.313725, 0.145098}
\definecolor{darkestpurple}{rgb}{0.396078, 0.113725, 0.196078}
\definecolor{greypurple}{rgb}{0.294118, 0.219608, 0.298039}
\definecolor{darkturquoise}{rgb}{0, 0.239216, 0.298039}
\definecolor{darkbrown}{rgb}{0.305882, 0.211765, 0.160784}
\definecolor{midgreen}{rgb}{0.560784, 0.6, 0.243137}
\definecolor{darkred}{rgb}{0.576471, 0.152941, 0.172549}
\definecolor{darkpurple}{rgb}{0.313725, 0.027451, 0.470588}
\definecolor{darkestblue}{rgb}{0, 0.156863, 0.333333}
\definecolor{lightpurple}{rgb}{0.776471, 0.690196, 0.737255}
\definecolor{softgreen}{rgb}{0.733333, 0.772549, 0.572549}
\definecolor{offwhite}{rgb}{0.839216, 0.823529, 0.768627}
\definecolor{medgreen}{rgb}{0.15, 0.6, 0.15}
\definecolor{midred}{rgb}{0.80,0.3,0.3}

\usepackage{listings}
\definecolor{keywordcolor}{rgb}{0.0, 0.1, 0.6}   
\definecolor{tacticcolor}{rgb}{0.0, 0.1, 0.6}    
\definecolor{commentcolor}{rgb}{0.4, 0.4, 0.4}   
\definecolor{symbolcolor}{rgb}{0.0, 0.1, 0.6}    
\definecolor{sortcolor}{rgb}{0.1, 0.5, 0.1}      
\definecolor{attributecolor}{rgb}{0.7, 0.1, 0.1} 

\newcommand{\rtext}[1]{\textcolor{midred}{\texttt{#1}}}

\newtheorem{dfn}{Definition}
\newtheorem{prop}{Proposition}
\newtheorem{thm}{Theorem}

\title{Certified Knowledge Compilation \\ with Application to Formally Verified Model Counting}

\author{
  Randal E. Bryant, Wojciech Nawrocki, Jeremy Avigad, and Marijn J. H. Heule
}

\institute{
  Carnegie Mellon University\\
  Pittsburgh, Pennsylvania 15221, USA
}

\authorrunning{Bryant, Nawrocki, Avigad, and Heule}
\titlerunning{Certified Knowledge Compilation}

\newcommand{\lean}{Lean~4}

\begin{document}

\maketitle

\begin{abstract}

Computing many useful properties of Boolean formulas, such as their weighted or unweighted model count,
is intractable on general representations. It can become tractable when formulas are expressed in a
special form, such as the decision decomposable negation normal form (\decdnnf{})\@.
\emph{Knowledge compilation} is the process of converting a formula
into such a form. Unfortunately existing knowledge compilers provide no guarantee that their output correctly
represents the original formula, and therefore they cannot validate a model count, or any other computed value.

We present \emph{Partitioned-Operation Graphs} (POGs), a form that can
encode all
of the representations used by existing knowledge compilers.
We have designed  CPOG, a framework that can express proofs of equivalence between a
POG  and a Boolean formula in conjunctive normal form (CNF).

We have developed a program that generates POG representations from \decdnnf{}
graphs
produced by the state-of-the-art knowledge compiler
\dfour{}, as well as checkable CPOG proofs certifying that the output POGs
are equivalent to the input CNF formulas.  Our toolchain
for generating and verifying POGs scales to all but the largest
graphs produced by \dfour{} for formulas from a recent model counting
competition. Additionally, we have developed a formally verified CPOG
checker and model counter for POGs in the \lean{} proof assistant.
In doing so, we proved the soundness of our proof framework. These programs
comprise the first formally verified toolchain for weighted and unweighted
model counting.
\end{abstract}

\section{Introduction}

Given a Boolean formula, modern Boolean satisfiability (SAT) solvers can
find an assignment satisfying it or generate a proof that no
such assignment exists.  They have applications across a variety of
domains including computational mathematics, combinatorial
optimization, and the formal verification of hardware, software, and
security protocols.  Some applications, however, require going
beyond Boolean satisfiability.  For example, the \emph{model
  counting problem} requires computing the number of satisfying
assignments of a formula, including in cases where there are far too many
to enumerate individually.  Model counting has
applications in artificial intelligence, computer security, and
statistical sampling.  There are also many useful extensions of model counting,
including {\em
  weighted model counting}, where a weight is defined for
each possible assignment, and the goal becomes to compute the sum of the weights
of the satisfying assignments.

Model counting is a challenging problem---more challenging than the
already NP-hard Boolean satisfiability.  Several
tractable variants of Boolean satisfiability, including 2-SAT, become
intractable when the goal is to count models and not just determine
satisfiability \cite{valiant:siam:1979}.  Nonetheless, a number of
model counters that scale to very large formulas have been developed, as
witnessed by the progress in recent model counting competitions.

One approach to model counting, known as \emph{knowledge compilation},
transforms the formula into a structured form for which model counting
is straightforward.  For example, the \emph{deterministic decomposable negation normal form}
(\detdnnf{}) introduced by
Darwiche~\cite{darwiche:jacm:2001}, as well as the more restricted
\emph{decision decomposable negation normal form} (\decdnnf{})~\cite{huang:jair:2007,beame:uai:2013}
represent a
Boolean formula as a directed acyclic graph, with terminal nodes
labeled by Boolean variables and their complements, and with each
nonterminal node labeled by a Boolean \textsc{and} or \textsc{or} operation.  Restrictions
are placed on the structure of the graph (described in Section~\ref{sect:pog}) such that a count of the
models can be computed by a single bottom-up traversal
Kimmig, et al.~\cite{kimmig:jal:2017}
present a very general {\em
  algebraic model counting} framework describing
properties of Boolean functions that can be efficiently computed from
a \detdnnf{} representation.  These include unweighted and weighted model
counting, and much more.

One shortcoming of existing knowledge compilers is that they have no
generally accepted
way to validate that
the compiled representation is logically equivalent to the original
formula.  By contrast, all modern SAT solvers can generate
checkable proofs when they encounter unsatisfiable formulas.  The
guarantee provided by a checkable certificate of correctness enables
users of SAT solvers to fully trust their results.  Experience has also
shown that being able to generate proofs allow SAT solver developers to quickly
detect and diagnose bugs in their programs. This, in turn, has led
to more reliable SAT solvers.

This paper introduces \emph{Partitioned-Operation Graphs} (POGs),
a form that can encode all of the representations produced by current knowledge
compilers. The CPOG (for ``certified'' POG) file format then
captures both the structure of a POG
and a checkable proof of its logical equivalence to a Boolean formula in
conjunctive normal form (CNF).  A CPOG
proof consists of a sequence of clause addition and deletion steps,
based on an extended resolution proof system~\cite{Tseitin:1983}.
We establish a set of conditions that, when satisified by a CPOG file, guarantees that it
encodes a well-formed POG and provides a valid equivalence proof.

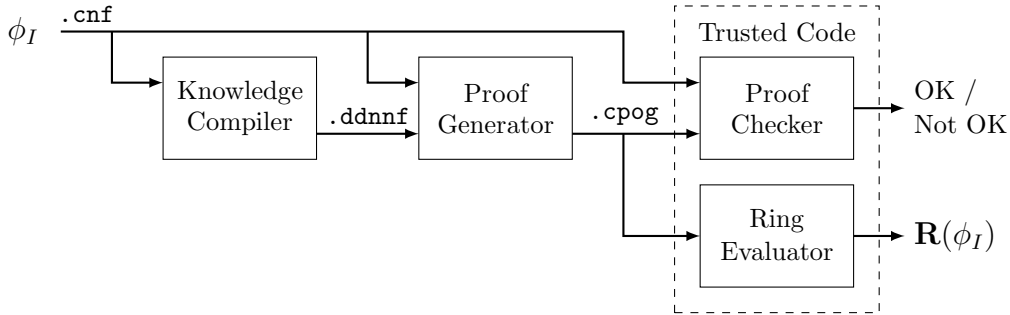
\begin{figure}
\centering{
\begin{tikzpicture}[scale=0.17]
  \draw (8,10) rectangle (20,18);
  \node at (14,15.2) {Knowledge};
  \node at (14,12.8) {Compiler};

  \draw (28,10) rectangle (40,18);
  \node at (34,15.2) {Proof};
  \node at (34,12.8) {Generator};

  \draw (50,10) rectangle (62,18);
  \node at (56,15.2) {Proof};
  \node at (56,12.8) {Checker};

  \draw (50,0) rectangle (62,8);
  \node at (56,5.2) {Ring};
  \node at (56,2.8) {Evaluator};

  \draw[dashed] (48,-2) rectangle (64,22);
  \node at (56,20) {Trusted Code};

  \draw[thick] (0,20) -- (44,20) -- (44,16) [-latex] -- (50,16);
  \draw[thick] (4,20) -- (4,16) [-latex] -- (8,16);
  \draw[thick] (24,20) -- (24,16) [-latex] -- (28,16);
  \node [left] at (-1,20) {\large {$\inputformula$}};
  \node [above] at (2,20) {\texttt{.cnf}};

  \draw[thick] (20,12) [-latex] -- (28,12);
  \node [above] at (24,12) {\texttt{.ddnnf}};

  \draw[thick] (40,12) [-latex] -- (50,12);
  \draw[thick] (44,12) -- (44,4) [-latex] -- (50,4);
  \node [above] at (44,12) {\texttt{.cpog}};

  \draw[thick] (62,14) [-latex] -- (66,14) ;
  \node [right] at (66,15.2) { OK /} ;
  \node [right] at (66,12.8) { Not OK} ;
  \draw[thick] (62,4) [-latex] -- (66,4) ;
  \node [right] at (66,4) {\large {$\textbf{R}(\inputformula)$}};

\end{tikzpicture}
}
\caption{Certifying toolchain.
  The ring evaluator produces a weighted or unweighted count.
  Certification by the proof checker guarantees its correctness.}
\label{fig:chain}
\end{figure}

Figure~\ref{fig:chain} illustrates our certifying knowledge
compilation and model counting toolchain.  Starting with input formula
$\inputformula$, the \dfour{} knowledge
compiler~\cite{lagniez:ijcai:2017} generates a \decdnnf{}
representation, and the \emph{proof generator} uses this to generate a
CPOG file.  The \emph{proof checker} verifies the equivalence of the
CNF and CPOG representations.  The \emph{ring evaluator} computes an
unweighted or weighted model count from the POG\@ representation.  As
the dashed box in Figure~\ref{fig:chain} indicates, this toolchain
moves the root of trust away from the complex and highly optimized
knowledge compiler to a relatively simple checker and evaluator.
Importantly, the proof generator need not be trusted---its errors will
be caught by the proof checker.  

To ensure soundness of the abstract CPOG proof system, as well as
correctness of its concrete implementation, we formally verified the
proof system as well as versions of the proof checker and ring
evaluator in the \lean{} proof assistant~\cite{demoura:cade:2021}.
Running these two programs on a  CPOG file gives strong
assurance that the proof and the model count are correct. Our
experience with developing a formally verified proof checker has shown
that, even within the well-understood framework of extended
resolution, it can be challenging to formulate a full set of
requirements that guarantee soundness.  In fact, as described in Section~\ref{sect:lean:subtle-condition}, our efforts to
formally verify our proof framework exposed subtle conditions that we had
to impose on our partitioned sum rule.  

We evaluate our toolchain using benchmark formulas from the 2022
unweighted and weighted model competitions.  Our tools handle all but
the largest graphs generated by \dfour{}.  We evaluate
the benefits of several optimizations, finding that the use of lemmas
to exploit the sharing of subgraphs in the \decdnnf{} representation can be critical to avoid an expansion of the graph into a tree.
We measure
the relative
performance of the verified checker with one designed for high
performance and capacity, finding that the time to run the verified checker remains within a factor of $4\times$ that of the high capacity checker for most benchmarks, and that it has similar scaling properties.
We also show that our tools can provide end-to-end verification of
formulas that have been transformed by an equivalence-preserving
preprocessor.  That is, verification is based on the original formula,
and so proof checking certifies correct operation of the preprocessor,
the knowledge compiler, and the proof generator.

Our current tool can only handle the representations
generated by the \dfour{} knowledge compiler, and it only supports a
subclass of the Boolean function properties enabled by algebraic model
counting~\cite{kimmig:jal:2017}.  Both of these shortcomings can be overcome by modest extensions, as is discussed in Section~\ref{sect:extensions}.

This paper is an extended version of one published at the 2023
Conference on the Theory and Application of Boolean
Satisfiability~\cite{bryant:sat:2023}.  It provides much greater detail about
the algorithms, the formal verification, and the experimental results.

\section{Related Work}
\label{sect:related}

Generating proofs of unsatisfiability in SAT solvers has a long
tradition~\cite{ZhangMalik} and has become widely accepted due to the
formulation of clausal proof systems for which proofs can readily be
generated and efficiently checked
\cite{heule:cade:2013,wetzler14_drattrim}.
A number of formally verified checkers have been developed within different verification frameworks~\cite{cruz-cade-2017,lrat,Lammich:20,Tan:2021}.
The associated proofs add clauses while preserving satisfiability until the empty clause is derived.
Our work builds on the well-established technology and tools associated with clausal proof systems,
but we require features not found in proofs of unsatisfiability.
In particular, our checker constructs an entirely new representation of the input formula.  The proof must demonstrate
that the new representation
satisfies a set of rules, and that it
is logically equivalent to the input formula.  This requires verifying
additional proof steps, including clause deletion steps, and subtle
invariants, as described in Sections~\ref{sect:cpog} and
\ref{sect:formally-verified-toolchain}.

Capelli, et al.~\cite{capelli:sat:2019,capelli:aaai:2021} developed a knowledge compiler that
generates a certificate in a proof system that is itself based on
\decdnnf{}\@.  Their \cdfour{}
program, a modified version of \dfour{}, generates annotations to the
compiled representation, providing information about how the compiled
version relates to the input clauses.  It also generates a file of
clausal proof steps in the DRAT format~\cite{wetzler14_drattrim}.
Completing the certification involves running two different checkers
on the annotated \decdnnf{} graph and the DRAT file.  Although the
authors make informal arguments regarding the soundness of their
frameworks, these do not provide strong levels of assurance.  Indeed,
we have identified a weakness in their methodology due to an
invalid assumption about the guarantees provided by \dtrim{}, the
program it uses to check the DRAT file.  This weakness is
\emph{exploitable}: their framework can be ``spoofed'' into accepting
an incorrect compilation.

In more detail, \cdfour{} emits a sequence of clauses $R$ that
includes the conflict clauses that arose during a top-down processing of
the input clauses.  Given input formula $\phi_I$, their first task is to
check whether $\phi_I \imply R$, i.e., that any assignment that satisfies $\phi_I$ also satisfies each of the clauses in $R$.
They then base other parts of their proof on that
property and use a separate program to perform a series of additional
checks.  They use \dtrim{} to prove the implication, checking that each clause in $R$
satisifies the \emph{resolution asymmetric tautology} (RAT) property with
respect to the preceding
clauses~\cite{jarvisalo:ijcar:2012,heule:cade:2013}.  Adding a RAT
clause $C$ to a set of clauses maintains satisfiability, i.e., it will not cause a satisfiable formula to become unsatisfiable.
On the other hand,
it does not necessarily
preserve models, i.e., it can exclude some previous satisfying assignments.  As an example, consider the following formulas over the variables $x_1$, $x_2$, and $x_3$:
\begin{center}
  \begin{tabular}{lccc}
    $\phi_1$: & $(x_1 \lor x_3)$ & & \\
    $\phi_2$: & $(x_1 \lor x_3)$ & $\land$ & $(x_2 \lor \obar{x}_3)$\\
  \end{tabular}
\end{center}
Clearly, these two formulas are not equivalent---$\phi_1$ has six
models, while $\phi_2$ has four.  In particular, $\phi_1$ allows
arbitrary assignments to variable $x_2$, while $\phi_2$ does not.  Critically, however, the
second clause of $\phi_2$ is RAT with respect to the first clause (i.e., $\phi_1$)---any
satisfying assignment to $\phi_1$ can be transformed into one that
also satisfies $\phi_2$ by setting $x_2$ to 1, while keeping the values for other variables fixed.

This weakness would allow a buggy (or malicious) version of \cdfour{}
to spoof the checking framework.  Given formula $\phi_1$ as input, it
could produce a compiled result, including annotations, based on $\phi_2$ and also include the second clause of $\phi_2$ in $R$.
The check with \dtrim{} would pass, as would the other tests performed
by their checker.  We have confirmed this possibility with
their compiler and checker.\footnote{Downloaded May 18, 2023 as\\
\url{https://github.com/crillab/d4/tree/333370cc1e843dd0749c1efe88516e72b5239174}.}

This weakness can be corrected by restricting \dtrim{} to only allow adding
clauses that obey the stronger \emph{reverse unit propagation} (RUP) property~\cite{goldberg,vangelder08_verifying_rup_proofs}.  Adding a RUP clause $C$ to a set of clauses
does not change the set of satisfying assignments.
We have added a
command-line argument to \dtrim{} that enforces this
restriction.\footnote{Available at
\url{https://github.com/marijnheule/drat-trim/releases/tag/v05.22.2023}.}  This weakness, however, illustrates the general challenge of
developing a new proof framework.
As we can attest,
without engaging in an effort to formally verify the framework, there are likely to be
conditions that make the framework unsound.

Fichte, et al.~\cite{fichte:sat:2022} devised the \mice{}
proof framework for model counting programs.  Their proof rules are
based on the algorithms commonly used by model counters.  They
developed a program that can generate proof traces from \decdnnf{}
graphs and a program to check adherence to their proof rules.  This
framework is not directly comparable to ours, since it only certifies
the unweighted model count, but it has similar goals.
Again, they provide only  informal arguments
regarding the soundness of their framework.

Both of these prior certification frameworks are strongly tied to the
algorithms used by the knowledge compilers and model counters.  Some
of the conditions to be checked are relevant only to specific
implementations.    Our framework is very general and is based on a small set
of proof rules.  It builds on the highly developed
concepts of clausal proof systems.  These factors were important in enabling formal verification.
In Section~\ref{sect:experimental},
we also compare the performance of our toolchain to these other two.  We find that the $\cdfour{}$ toolchain generally
outperforms ours, while the \mice{} toolchain does not scale as well, especially for \decdnnf{} graphs with extensive sharing among the subgraphs.

\section{Logical Foundations}
\label{sect:logical:foundations}

  Let $\varset$ denote a set of Boolean variables, and let $\assign$
  be an \emph{assignment} of truth values to some subset of the
  variables, where $0$ denotes false and $1$ denotes true, i.e.,
  $\assign \colon \varset' \rightarrow \{0,1\}$ for some $\varset'
  \subseteq \varset$.  We say the assignment is \emph{total} when it
  assigns a value to every variable ($\varset' = \varset$), and that
  it is \emph{partial} otherwise.
  The set of all possible total assignments over
  $\varset$ is denoted $\uassign$.

For each variable $x \in \varset$,
  we define the \emph{literals} $x$ and $\obar{x}$, where $\obar{x}$ is the
  negation of $x$. An
  assignment $\assign$ can be viewed as a set of literals, where
  we write $\lit \in \assign$ when $\lit = x$ and $\assign(x) = 1$ or when
  $\lit = \obar{x}$ and $\assign(x) = 0$.  We write the negation of literal $\lit$ as $\obar{\lit}$.  That is, $\obar{\lit} = \obar{x}$ when $\lit = x$ and
$\obar{\lit} = x$ when $\lit = \obar{x}$.

\begin{dfn}
  The set of Boolean formulas is defined recursively.  Each
  formula $\phi$ has an associated \emph{dependency set}
  $\dependencyset(\phi)  \subseteq \varset$, and a set of \emph{models} $\modelset(\phi)$,
  consisting of total assignments that satisfy the formula:
  \begin{enumerate}
  \item Boolean constants $0$ and $1$ are Boolean formulas,
    with $\dependencyset(0) = \dependencyset(1) = \emptyset$, with $\modelset(0) = \emptyset$, and with $\modelset(1) = \uassign$.
  \item Variable $x$ is a Boolean formula, with $\dependencyset(x) = \{x\}$
    and $\modelset(x) = \{\assign \in \uassign | \assign(x)=1\}$.
  \item For formula $\phi$, its \emph{negation}, written $\boolnot \phi$ is a Boolean formula,
    with $\dependencyset(\boolnot \phi) = \dependencyset(\phi)$ and $\modelset(\boolnot \phi) = \uassign - \modelset(\phi)$.
  \item For formulas $\phi_1, \phi_2, \ldots, \phi_k$, their \emph{product} $\phi = \bigwedge_{1 \leq i \leq k} \phi_i$ is a Boolean formula, with
      $\dependencyset(\phi) = \bigcup_{1 \leq i \leq k} \dependencyset(\phi_i)$ and
      $\modelset(\phi) = \bigcap_{1 \leq i \leq k} \modelset(\phi_i)$.
  \item For formulas $\phi_1, \phi_2, \ldots, \phi_k$, their \emph{sum} $\phi = \bigvee_{1 \leq i \leq k} \phi_i$ is a Boolean formula, with
      $\dependencyset(\phi) = \bigcup_{1 \leq i \leq k} \dependencyset(\phi_i)$ and
      $\modelset(\phi) = \bigcup_{1 \leq i \leq k} \modelset(\phi_i)$.
  \end{enumerate}
\label{def:boolean}
\end{dfn}

  We highlight some special classes of Boolean formulas.  A formula is
  in \emph{negation normal form} (NNF) when negation is applied only to
  variables.  A formula is in \emph{conjunctive normal form} (CNF)
  when (i) it is in negation normal form, (ii) sum is applied only
  to literals, and (iii) there is a single product operation over all
  of the sums.  A CNF formula can be represented as a set of
  \emph{clauses}, each of which is a set of literals.  Each clause
  represents the sum of the literals, and the formula is the product
  of its clauses.  We use set notation to reference the clauses within
  a formula and the literals within a clause.  A clause consisting of
  a single literal is referred to as a \emph{unit} clause and the
  literal as a \emph{unit} literal.  This literal must be assigned
  value $1$ by any satisfying assignment of the formula.

\begin{dfn}
  \label{def:partitioned-operation-formula}
  A \emph{partitioned-operation formula}
 satisfies the following for all product and sum operations:
      \begin{enumerate}
      \item The arguments to each product must have disjoint dependency sets.  That is, operation
        $\bigwedge_{1 \leq i \leq k} \phi_i$ requires $\dependencyset(\phi_i) \cap \dependencyset(\phi_j) = \emptyset$ for $1 \leq i < j \leq k$.
      \item The arguments to each sum must have disjoint models.  That is, operation
        $\bigvee_{1 \leq i \leq k} \phi_i$ requires $\modelset(\phi_i) \cap \modelset(\phi_j) = \emptyset$ for $1 \leq i < j \leq k$.
      \end{enumerate}
\end{dfn}
     We let $\pand$ and $\por$ denote the product and sum operations
     in a partitioned-operation formula.  In the knowledge compilation
     literature, Boolean formulas where all product arguments have
     disjoint dependency sets are said to be
     \emph{decomposable}~\cite{darwiche:jacm:2001,darwiche:jair:2002}.  Those where all sum arguments have disjoint models are said to be
     \emph{deterministic}~\cite{darwiche:aaai:2002,darwiche:jair:2002}.

  \section{Ring Evaluation of a Boolean Formula}
\label{sect:ring}

We propose a general framework for summarizing properties of Boolean
formulas similar to the formulation of algebraic model counting by Kimmig, et al.~\cite{kimmig:jal:2017}.
Our formulation in terms of rings is more restrictive than their semiring-based approach.
We discuss the difference and how our
work could be generalized in Section~\ref{sect:extend:semiring}.

\begin{dfn}
  A \emph{commutative ring} $\ring$ is an algebraic structure
  $\langle \dset, \radd, \rmul, \addident, \mulident \rangle$,
  with elements in the set $\dset$ and with commutative and
  associative operations $\radd$ (addition) and $\rmul$ (multiplication),
  such that multiplication distributes
  over addition.  $\addident$ is the additive identity and $\mulident$ is
  the multiplicative identity.  Every element $a \in \dset$ has an
  \emph{additive inverse} $-a$ such that $a + -a = \addident$.
\label{def:ring}
\end{dfn}
We write $a - b$ as a shorthand for $a + -b$.

\begin{dfn}[Ring Evaluation Problem]
\label{def:ring_evaluation}
  For commutative ring $\ring$, a \emph{ring weight function} associates a value $w(x) \in \dset$ with
  every variable $x \in \varset$.  We then define $w(\obar{x}) \doteq \mulident-w(x)$.

  For Boolean formula $\phi$ and ring weight function $w$, the \emph{ring evaluation problem} computes
  \begin{equation}
    \begin{array}{rcl}
    \rep(\phi, w) & = & \sum_{\alpha \in \modelset(\phi)} \;\; \prod_{\lit \in \alpha} w(\ell) \label{eqn:rep}
    \end{array}
  \end{equation}
  In this equation, sum \scalebox{0.8}{$\sum$} is computed using addition operation $\radd$, and product \scalebox{0.8}{$\prod$} is computed using multiplication operation $\rmul$.
\label{def:weight}
\end{dfn}

Many important properties of Boolean formulas can be
expressed as ring evaluation problems.  The
(unweighted) \emph{model counting} problem for formula $\phi$ requires determining $|\modelset(\phi)|$.
It can be cast as a ring evaluation problem by having $\radd$ and
$\rmul$ be addition and multiplication over rational numbers and using
weight function $w(x) = 1/2$ for every variable $x$.
Ring evaluation of formula $\phi$ gives the \emph{density} of
the formula, i.e., the fraction of all possible total assignments that are
models.  For $n = |\varset|$, scaling the density by $2^n$
yields the number of models.

The \emph{weighted model counting}  problem is also defined over
rational numbers.  Some formulations  allow
independently assigning weights $W(x)$ and $W(\obar{x})$ for each variable $x$ and its complement, with the possibility that
$W(x) + W(\obar{x}) \not = 1$.
We can cast this as a
ring evaluation problem by letting $r(x) = W(x) + W(\obar{x})$,
performing ring evaluation with weight function $w(x) = W(x)/r(x)$ for each
variable $x$, and computing the weighted count
as $\rep(\phi, w)\; \rmul\; \prod_{x \in \varset} r(x)$.
Of course, this requires that $r(x) \not = 0$ for all $x \in \varset$.

The \emph{function hashing problem} provides a test
of inequivalence for Boolean formulas.  That is, for $n = |\varset|$, let $\ring$ be a
finite  field with $|\dset| = m$ such that $m \geq 2 n$.  For each $x \in \varset$, choose a value from $\dset$ at random for $w(x)$.  Two formulas
$\phi_1$ and $\phi_2$ will clearly have $\rep(\phi_1, w) = \rep(\phi_2, w)$
if they are logically equivalent,
and if $\rep(\phi_1, w) \not = \rep(\phi_2, w)$, then they are clearly inequivalent.
If they are not equivalent, then
the probability that $\rep(\phi_1, w) \not = \rep(\phi_2, w)$ will be at
least $\left(1-\frac{1}{m}\right)^n \geq \left(1-\frac{1}{2n}\right)^n > 1/2$.
Function hashing can therefore be used as part of a
randomized algorithm for equivalence testing~\cite{blum:ipl:1980}.
For example, it can compare different runs on a single formula,
either from different compilers or from a single compiler with different configuration parameters.

\section{Partitioned-Operation Graphs (POGs)}
\label{sect:pog}

Performing ring evaluation on an arbitrary Boolean formula could be intractable, but it is straightforward for a formula with partitioned operations:
\begin{prop}
\label{prop:ring:eval}
Ring evaluation with operations $\boolnot$, $\pand$, and $\por$ satisfies the following for any weight function $w$:
\begin{eqnarray}
\rep(\boolnot \phi,\; w) &=& \mulident - \rep(\phi, w) \label{eqn:ring:negation}\\
\textstyle
\scalebox{1.2}{\rep}\left(\Pand_{1 \leq i \leq k} \phi_i,\; w \right) &=& \prod_{1 \leq i \leq k} \rep(\phi_i, w) \label{eqn:ring:product}\\
\textstyle
\scalebox{1.2}{\rep}\left(\Por_{1 \leq i \leq k} \phi_i,\; w\right) &=& \sum_{1 \leq i \leq k} \rep(\phi_i, w) \label{eqn:ring:sum} 
\end{eqnarray}
\end{prop}
As is described in \ref{sect:formally-verified-toolchain}, we have proved these three equations using \lean{}.

A \emph{partitioned-operation graph} (POG) is a directed, acyclic
graph with nodes $N$ and edges $E \subseteq N \times N$.  We denote
nodes with boldface symbols, such as $\nodeu$ and $\nodev$.  When
$(\nodeu,\nodev) \in E$, node $\nodev$ is said to be a \emph{child} of
node $\nodeu$.  The in- and out-degrees of node $\nodeu$ are defined
as $\indegree(\nodeu) = | E \cap (N \times \{\nodeu\}) |$, and
$\outdegree(\nodeu) = | E \cap (\{\nodeu\} \times N) |$.  Node
$\nodeu$ is said to be \emph{terminal} if $\outdegree(\nodeu) = 0$.  A
terminal node is labeled by a Boolean constant or variable.  Node
$\nodeu$ is said to be \emph{nonterminal} if $\outdegree(\nodeu) > 0$.
A nonterminal node is labeled by Boolean operation $\pand$ or $\por$.
A node can be labeled with operation $\pand$ or $\por$ only if it
satisfies the partitioning restriction for that operation.  Every POG
has a designated \emph{root node} $\noder$.  Each edge has
a \emph{polarity}, indicating whether (negative polarity) or not
(positive polarity) the corresponding argument should be negated.

A POG represents a partitioned-operation
formula with a sharing of common subformulas.  Every node in the graph can be viewed as a partitioned-operation formula, and so we write
$\phi_{\nodeu}$ as the formula denoted by node $\nodeu$.
Each such formula has a set of models $\modelset(\phi_{\nodeu})$.

We can now define and compare two related representations:
\begin{itemize}
\item A \detdnnf{} graph can be viewed as a POG with negation applied only to variables.
\item A \decdnnf{} graph is a \detdnnf{} graph with the further
  restriction that any sum node $\nodeu$ has exactly two children
  $\nodeu_1$ and $\nodeu_0$, and it has an associated \emph{decision variable} $x$.
  For $b \in \{0,1\}$, node $\nodeu_b$ can be a terminal node with variable $x$, where the polarity of the edge from $\nodeu$ to $\nodeu_b$ is
  negative for $b=0$ and positive for $b=1$.
  Alternatively, $\nodeu_b$ can be a product node having either literal $\obar{x}$ ($b=0$), or literal $x$ ($b=1$) as one of its arguments.
  Either form implies that any total assignment $\assign \in \modelset(\phi_{\nodeu_b})$
  has $\assign(x)=b$, for $b \in \{0,1\}$.
\end{itemize}
The generalizations encompassed by POGs have also been referred to as \emph{deterministic decomposable circuits} (d-Ds)~\cite{monet:amw:2018}.
Our current proof generator only works for knowledge compilers
generating \decdnnf{} representations, but these generalizations
allow for future extensions, while maintaining the ability to
efficiently perform ring evaluation.  Extending the tool to handle arbitrary POGs is discussed in Section~\ref{sect:extend:pog}.

We define the \emph{size} of POG $P$, written $|P|$, to be the
the number of nonterminal nodes plus the number of edges from these nodes to their children.  Ring
evaluation of $P$ can be performed with at most $|P|$ ring
operations by traversing the graph from the terminal nodes up to
the root, computing a value $\rep(\phi_\nodeu, w)$ for each node $\nodeu$.
The final result is then $\rep(\phi_\noder, w)$.

\section{Clausal Proof Framework}
\label{sect:clausal}

A proof in our framework consists of a sequence of clause addition and
deletion steps, with each step preserving the set of solutions to the
original formula.  The state of the proof at any step is represented
as a set of \emph{active} clauses $\theta$, i.e., those that have been
added but not yet deleted.  Our framework is based on \emph{extended}
resolution~\cite{Tseitin:1983}, where proof steps can introduce new
\emph{extension variables} encoding Boolean formulas over input and
prior extension variables.  That is, each extension variable $z$ is
introduced via a set of \emph{defining clauses} that encode a formula
$z \ifandonlyif F$, where $F$ is a Boolean formula over a subset of the input
variables $\varset$ and previously defined extension variables $\exvarset$.
We write $\theta$ for
formulas encoded as clauses, possibly with extension variables, and
$\phi$ for formulas
that use no extension variables.

Let $\exvarset(\theta)$
denote the set of extension variables occurring in formula $\theta$.
For any total assignment $\assign$ to the variables in $\varset$, the defining
clauses induce a unique assignment $\eassign$ to the variables in $\varset
\cup \exvarset(\theta)$.  For Boolean formula $\phi$ over variables
$\varset$ and clausal
formula $\theta$ over the variables $\varset \cup \exvarset(\theta)$, we say that $\phi$ is {\em
  equivalent over $\varset$} to $\theta$, written
$\phi \eequiv \theta$,
when for any assignment
$\assign$ to the variables in $\varset$, assignment $\assign$ is a model of $\phi$ if
and only if its extension $\eassign$ is a model of $\theta$.
Starting with  $\theta$ equal to input formula $\inputformula$,
the proof must maintain the invariant that
$\inputformula \eequiv \theta$.

Clauses can be added in two different ways.  One is when they serve as
the defining clauses for an extension variable.  This form
occurs only when defining $\pand$ and $\por$ operations, as is
described in Section~\ref{sect:cpog}.  Clauses can also be added or
deleted based on \emph{implication redundancy}.  That is, when clause
$C$ satisfies $\theta \imply C$ for formula $\theta$, then it can either
be added to $\theta$ to create the formula $\theta \cup \{C\}$ or it can be deleted
from $\theta \cup \{C\}$ to create $\theta$.

We use \emph{reverse unit propagation} (RUP) to certify
implication redundancy when adding or deleting
clauses~\cite{goldberg,vangelder08_verifying_rup_proofs}.
RUP
is the core rule supported by standard
proof checkers~\cite{heule:cade:2013,wetzler14_drattrim} for propositional logic. It provides a simple and efficient
way to check a sequence of applications of the resolution proof rule~\cite{robinson-1965}.
Let $C = \{\lit_1, \lit_2, \ldots,\lit_p\}$ be a clause to be
proved redundant with respect to formula $\theta$.  Let $D_1, D_2, \ldots, D_k$ be a sequence of supporting
\emph{antecedent} clauses, such that each $D_i$ is in $\theta$.
A RUP step
proves that $\bigwedge_{1\leq i \leq k} D_i \imply C$ by showing
that the combination of the antecedents plus the negation of $C$ leads
to a contradiction.  The negation of $C$ is the formula
$\overline{\lit}_1 \land \overline{\lit}_2 \land \cdots \land
\overline{\lit}_p$, having a CNF representation consisting of $p$ unit
clauses of the form $\obar{\lit}_i$ for $1 \leq i \leq p$.  A RUP
check processes the clauses of the antecedent in sequence, inferring
additional unit clauses.  In processing clause $D_i$, if all but one
of the literals in the clause is the negation of one of the
accumulated unit clauses, then we can add this literal to the
accumulated set.  That is, all but this literal have been falsified,
and so it must be set to true for the clause to be satisfied.  The
final step with clause $D_k$ must cause a contradiction, i.e., all of
its literals are falsified by the accumulated unit clauses.

Compared to the proofs of unsatisfiability generated by SAT solvers,
ours have important differences.  Most
significantly, each proof step must preserve the set of solutions with respect to the input variables;
our proofs must therefore justify both clause deletions and additions.
By contrast, an unsatisfiability proof need only guarantee that
no proof step causes a satisfiable set of clauses to become
unsatisfiable, and therefore it need only justify clause additions.

\section{The CPOG Representation and Proof System}
\label{sect:cpog}

A CPOG file provides both a declaration of a POG, as well as a checkable
proof that a Boolean formula, given in conjunctive normal
form, is logically equivalent to the POG\@.
The proof format draws its inspiration from the LRAT~\cite{lrat} and
QRAT~\cite{heule:JAR2014} formats for unquantified and quantified Boolean formulas, respectively.
Key properties include:
\begin{itemize}
  \item
  The file contains declarations of $\pand$ and $\por$ operations to describe the POG.
  Declaring a node $\nodeu$ implicitly adds an extension variable $u$ and a set of defining clauses $\theta_{u}$
  encoding the product or sum operation.
  This is the only means for adding extension variables to the proof.
\item Boolean negation is supported implicitly by allowing the
  arguments of the $\por$ and $\pand$ operations to be literals and not just
  variables.
\item
  The file contains explicit clause addition steps.
  A clause can only be added if it is logically implied by the existing clauses.
  A sequence of clause identifiers must be listed as a \emph{hint} providing a RUP verification of the implication.
\item
  The file contains explicit clause deletion steps.
  A clause can only be deleted if it is logically implied by the remaining clauses.
  A sequence of clause identifiers must be listed as a \emph{hint} providing a RUP verification of the implication.
\item The checker must track the dependency set for every input and
  extension variable.  For each $\pand$ operation, the checker must ensure that the dependency sets for its arguments are disjoint.
  The associated extension variable has a dependency set equal to the union of those of its arguments.
\item Declaring a $\por$ operation requires a sequence of clauses
  providing a RUP proof that the arguments are mutually exclusive.
  Only binary $\por$ operations are allowed to avoid requiring multiple proofs of disjointness.
\end{itemize}

\subsection{Syntax}
\label{subsection:syntax}

\begin{table}
  \caption{CPOG Step Types.  $C$: clause identifier, $L$: literal, $V$: variable}
  \label{tab:cpog:syntax}
\centering{
  \begin{tabular}{lllll}
    \toprule
    \multicolumn{4}{c}{Rule} & \multicolumn{1}{c}{Description} \\
    \midrule
    \makebox[5mm][l]{$C$} & \makebox[10mm][l]{\texttt{a}}  & \makebox[15mm][l]{$L^{*}$ \texttt{0}} & \makebox[15mm][l]{$C^{+}$ \texttt{0}}  & \makebox[20mm][l]{Add RUP clause} \\
     & \texttt{d} & $C$             & $C^{+}$  \texttt{0} & Delete RUP clause \\
    \midrule
    $C$    & \texttt{p} & $V \; L^{*}$ \texttt{0}    &                  & Declare $\pand$ operation \\
    $C$    & \texttt{s} & $V \; L \; L$    & $C^{+}$ \texttt{0}  & Declare $\por$ operation \\
    \midrule
     & \texttt{r} & $L$             &            & Declare root literal\\
    \bottomrule
  \end{tabular}
  }
\end{table}

Table~\ref{tab:cpog:syntax} shows the declarations that can occur in a CPOG file.
As with other clausal proof formats, a variable is
represented by a positive integer $v$, with the first ones being input
variables and successive ones being extension variables.  Literal $\lit$
is represented by a signed integer, with $-v$ being the logical negation of
variable $v$.  Each clause is indicated by a positive integer
identifier $C$, with the first ones being the IDs of the input clauses and successive
ones being the IDs of added clauses.  Clause identifiers must be defined in order,
with any clause identifier $C'$ given in the hint when adding clause $C$ having $C' < C$.

The first set of proof rules are similar to those in other clausal
proofs.
Clauses can be added via RUP addition
(command \texttt{a}), with a sequence of antecedent clauses (the
``hint'').
Similarly for clause deletion (command \texttt{d}).

\begin{table}
\caption{Defining Clauses for Product (A) and Sum (B) Operations}
\begin{minipage}{0.54\textwidth}
\begin{center}
\begin{tabular}{cccccc}
\multicolumn{6}{c}{(A) Product Operation $\pand$}\\
\toprule
\makebox[10mm]{ID} & \multicolumn{5}{c}{Clause} \\
\midrule
  $i$ & $v$ & $-\lit_1$ & $-\lit_2$ & $\cdots$ & $-\lit_k$\\
  $i\!+\!1$ & $-v$ & $\lit_1$  \\
  $i\!+\!2$ & $-v$ & $\lit_2$  \\
  & $\ldots$ \\
  $i\!+\!k$ & $-v$ & $\lit_k$  \\
\bottomrule
\end{tabular}
\end{center}
\end{minipage}
\begin{minipage}{0.44\textwidth}
\begin{center}
\begin{tabular}{cccc}
\multicolumn{4}{c}{(B) Sum Operation $\por$}\\
\toprule
\makebox[10mm]{ID} & \multicolumn{3}{c}{Clause} \\
\midrule
  $i$ & $-v$ & $\lit_1$ & $\lit_2$ \\
  $i\!+\!1$ & $v$ & $-\lit_1$ \\
  $i\!+\!2$ & $v$ & $-\lit_2$ \\
\bottomrule
$\;$ \\
$\;$ \\
\end{tabular}
\end{center}
\end{minipage}
\label{tab:defining}
\end{table}

The declaration of a \emph{product} operation, creating a node with operation $\pand$,
 has the form:
\begin{center}
\begin{tabular}{ccccccccc}
  \makebox[5mm]{$i$} & \makebox[5mm]{\texttt{p}} & \makebox[5mm]{$v$} & \makebox[5mm]{$\lit_1$} & \makebox[5mm]{$\lit_2$} &
  \makebox[5mm]{$\cdots$} & \makebox[5mm]{$\lit_k$} & \makebox[5mm]{\texttt{0}} \\
\end{tabular}
\end{center}
Integer $i$ is a new clause ID, $v$ is a positive integer that does not
correspond to any previous variable, and $\lit_1, \lit_2, \ldots, \lit_k$ is a sequence of $k$
integers, indicating the arguments as literals of existing variables.
As Table~\ref{tab:defining}(A) shows,
this declaration implicitly causes $k+1$ clauses to be added to the proof, providing a Tseitin encoding that defines extension variable $v$ as the product of its arguments.

The dependency sets for the arguments represented by each pair of
literals $\lit_i$
and $\lit_{j}$ must
be disjoint, for $1 \leq i < j \leq k$.  A product operation may have no arguments,
representing Boolean constant $1$.  The only clause added to the proof will be
the unit literal $v$.  A reference to literal $-v$ then provides a way
to represent constant $0$.

The declaration of a \emph{sum} operation, creating a node with operation $\por$, has the form:
\begin{center}
\begin{tabular}{ccccccc}
  \makebox[5mm]{$i$} & \makebox[5mm]{\texttt{s}} & \makebox[5mm]{$v$} & \makebox[5mm]{$\lit_1$} & \makebox[5mm]{$\lit_2$}
\makebox[5mm]{$H$} & \makebox[5mm]{$\texttt{0}$} \\
\end{tabular}
\end{center}
Integer $i$ is a new clause ID, $v$ is a positive integer that does
not correspond to any previous variable, and $\lit_1$ and $\lit_2$ are
signed integers, indicating the arguments as literals of existing variables.  Hint $H$
consists of a
sequence of clause IDs, all of which must be defining clauses for other POG operations.\footnote{The restriction to defining clauses in the hint is critical to soundness.
Allowing the hint to include the IDs of input clauses creates an exploitable weakness.  We discovered this weakness in the course of our efforts at formal verification.}
As Table~\ref{tab:defining}(B) shows,
this declaration implicitly causes three clauses to be added to the proof, providing a Tseitin encoding that defines extension variable $v$ as the sum of its arguments.
The hint must provide a RUP proof of the clause $\obar{\lit}_1 \lor \obar{\lit}_2$, showing that the two children of this node have disjoint models.

Finally, the literal denoting the root of the POG is declared with the
\texttt{r} command.  It can occur anywhere in the file.  Except in degenerate cases, it
will be the extension variable representing the root of a graph.

\subsection{Semantics}
\label{subsection:semantics}

As was described in Section~\ref{sect:clausal}, the defining clauses
for the product and sum operations uniquely define the values of their
extension variables for any assignment of values to the argument
variables.  That is, for assignment $\assign$ to the variables in
$\varset$, the defining clauses induce a unique assignment $\eassign$
to all data and extension variables.
Every POG node $\nodeu$ represents POG formula $\phi_{\nodeu}$ and has an associated extension variable $u$.
We can prove that
for any total assignment $\assign$ to the input variables, we will have
$\eassign(u) = 1$ if and only if $\assign \in \modelset(\phi_\nodeu)$.


The sequence of operator declarations, asserted clauses, and
clause deletions represents a systematic transformation of the input formula $\inputformula$
into a POG\@.  Validating all of these steps serves to prove that
POG $P$ is logically equivalent to the input formula.
At the completion of the proof, the following \textsc{final conditions} must hold:
\begin{enumerate}
\item There is exactly one remaining clause that was added via RUP
  addition, and this is a unit clause consisting of root literal $r$.
\item All of the input clauses have been deleted.
\end{enumerate}
In other words, at the end of the proof it must hold that the active
clauses be exactly those in $\pogformula \doteq \{\{r\}\} \cup \;
\bigcup_{\nodeu \in P} \theta_{u}$, the formula consisting of unit
clause $\{r\}$ and the defining clauses for the nodes, providing a
Tseitin encoding of $P$.  By our invariant, we are guaranteed that
$\inputformula \eequiv \pogformula$.  That is, for any total
assignment $\assign$ to the input variables,
$\assign$ is in $\modelset(\inputformula)$ if and only if
its unique extension $\eassign$ to the POG variables satisfies $\eassign(r) = 1$.

The sequence of clause addition steps provides a \emph{forward implication} proof that
$\assign \in \modelset(\inputformula) \imply \eassign(r) = 1$.  That is, any total
assignment $\assign$ satisfying the input formula must, when extended, also satisfy
the formula represented by the POG\@.
Conversely,
the sequence of clause deletion steps that delete all intermediate added clauses and all input clauses
provides a \emph{reverse implication} proof:
$\eassign(r) = 1 \imply \assign \in \modelset(\inputformula)$.
It does so by contradiction, proving that when $\eassign(r) = 0$, we must have $\assign \not \in \modelset(\inputformula)$.

\newcommand{\smallnl}{\\[-1pt]}

\subsection{CPOG Example}
\label{sect:cpog:example}

\begin{figure}
\vspace{-10pt}
\begin{minipage}{0.58\textwidth}
(A)  Input Formula\\[1.2ex]
\begin{tabular}{lll}
\toprule
\makebox[5mm]{ID} & \makebox[15mm]{Clauses} & \\
\midrule
\rtext{1} & \texttt{-1 3 -4} & \texttt{0} \smallnl
\rtext{2} & \texttt{-1 -3 4} & \texttt{0} \smallnl
\rtext{3} & \texttt{3 -4} & \texttt{0}\smallnl
\rtext{4} & \texttt{1 -3 4} & \texttt{0} \smallnl
\rtext{5} & \texttt{-1 -2} & \texttt{0} \\
\bottomrule
\end{tabular}
\\[1.8ex]
(C) POG Declaration\\[1.2ex]
\begin{tabular}{llll}
\toprule
\makebox[5mm]{ID} & \multicolumn{2}{l}{CPOG line} & Explanation \\
\midrule
\rtext{6} & \texttt{p 5 -3 -4} & \texttt{0} & $p_5 = \obar{x}_3 \pand \obar{x}_4$ \\
\rtext{9} & \texttt{p 6 3 4} & \texttt{0} & $p_6 = x_3 \pand x_4$ \\
\rtext{12} & \texttt{s 7 5 6} \; \rtext{7 10} & \texttt{0} & $s_7 = p_5 \por p_6$ \\
\rtext{15} & \texttt{p 8 -1 7} & \texttt{0} & $p_8 = \obar{x}_1 \pand s_7$ \\
\rtext{18} & \texttt{p 9 1 -2 7} & \texttt{0} & $p_9 = x_1 \pand \obar{x}_2 \pand s_7$ \\
\rtext{22} & \texttt{s 10 8 9} \; \rtext{16 19} & \texttt{0} & $s_{10} = p_8 \por p_9$ \\
 & \texttt{r 10} && Root $r = s_{10}$\\
\bottomrule
\end{tabular}
\end{minipage}
\begin{minipage}{0.35\textwidth}
(B) POG Representation \\
\begin{tikzpicture}
\definecolor{fillcolor}{RGB}{255,255,255}
\definecolor{highcolor}{RGB}{0,0,0}
\definecolor{lowcolor}{RGB}{0,0,0}
\definecolor{neutralcolor}{RGB}{0,0,0}
\definecolor{pathcolor}{RGB}{0,0,0}
\definecolor{background}{RGB}{225,225,225}
\draw (3.00,7.85) [thin,neutralcolor]  -- (3.00,6.53);
\draw (3.00,6.53) [thin,neutralcolor]  -- (2.12,4.77);
\draw (3.00,6.53) [thin,neutralcolor]  -- (3.88,4.77);
\draw (2.12,4.77) [thin,neutralcolor]  -- (3.00,3.00);
\draw (3.88,4.77) [thin,neutralcolor]  -- (3.00,3.00);
\draw (3.00,3.00) [thin,neutralcolor]  -- (2.12,1.24);
\draw (3.00,3.00) [thin,neutralcolor]  -- (3.88,1.24);
\draw (0.35,5.65) [thin,neutralcolor]  -- (3.00,5.65);
\draw (1.24,5.65) [thin,neutralcolor]  -- (2.12,4.77);
\draw (3.00,5.65) [thin,neutralcolor]  -- (3.88,4.77);
\draw (3.00,5.65) [thin,neutralcolor]  -- (3.00,5.65);
\draw [fill=neutralcolor,draw=neutralcolor] (3.00,5.65) circle [radius=0.09];
\draw (0.35,3.88) [thin,neutralcolor]  -- (1.24,3.88);
\draw (1.50,4.15) [thin,neutralcolor]  -- (1.50,4.15);
\draw [fill=neutralcolor,draw=neutralcolor] (1.50,4.15) circle [radius=0.09];
\draw (1.24,3.88) [thin,neutralcolor]  -- (2.12,4.77);
\draw (0.35,2.12) [thin,neutralcolor]  -- (3.00,2.12);
\draw (1.24,2.12) [thin,neutralcolor]  -- (2.12,1.24);
\draw (3.00,2.12) [thin,neutralcolor]  -- (3.88,1.24);
\draw (3.00,2.12) [thin,neutralcolor]  -- (3.00,2.12);
\draw [fill=neutralcolor,draw=neutralcolor] (3.00,2.12) circle [radius=0.09];
\draw (0.35,0.35) [thin,neutralcolor]  -- (3.00,0.35);
\draw (1.24,0.35) [thin,neutralcolor]  -- (2.12,1.24);
\draw (3.00,0.35) [thin,neutralcolor]  -- (3.88,1.24);
\draw (3.00,0.35) [thin,neutralcolor]  -- (3.00,0.35);
\draw [fill=neutralcolor,draw=neutralcolor] (3.00,0.35) circle [radius=0.09];
\draw [thin,fill=fillcolor,draw=fillcolor] (3.00,7.85) circle [radius=0.35];
\node at (3.00,7.85) {$\noder$};
\draw [thin,fill=fillcolor,draw=fillcolor] (0.35,5.65) circle [radius=0.35];
\node at (0.35,5.65) {$x_1$};
\draw [thin,fill=fillcolor,draw=fillcolor] (0.35,3.88) circle [radius=0.35];
\node at (0.35,3.88) {$x_2$};
\draw [thin,fill=fillcolor,draw=fillcolor] (0.35,2.12) circle [radius=0.35];
\node at (0.35,2.12) {$x_3$};
\draw [thin,fill=fillcolor,draw=fillcolor] (0.35,0.35) circle [radius=0.35];
\node at (0.35,0.35) {$x_4$};
\draw [thin,fill=fillcolor,draw=neutralcolor] (3.00,6.53) circle [radius=0.44];
\node at (3.00,6.53) {$\lor^{\textsf{p}}$};
\draw [thin,fill=fillcolor,draw=neutralcolor] (2.12,4.77) circle [radius=0.44];
\node at (2.12,4.77) {$\land^{\textsf{p}}$};
\draw [thin,fill=fillcolor,draw=neutralcolor] (3.88,4.77) circle [radius=0.44];
\node at (3.88,4.77) {$\land^{\textsf{p}}$};
\draw [thin,fill=fillcolor,draw=neutralcolor] (3.00,3.00) circle [radius=0.44];
\node at (3.00,3.00) {$\lor^{\textsf{p}}$};
\draw [thin,fill=fillcolor,draw=neutralcolor] (2.12,1.24) circle [radius=0.44];
\node at (2.12,1.24) {$\land^{\textsf{p}}$};
\draw [thin,fill=fillcolor,draw=neutralcolor] (3.88,1.24) circle [radius=0.44];
\node at (3.88,1.24) {$\land^{\textsf{p}}$};
\node at (1.24,5.65) {};
\node at (3.00,5.65) {};
\node at (3.00,5.65) {};
\node at (1.24,3.88) {};
\node at (1.50,4.15) {};
\node at (1.24,2.12) {};
\node at (3.00,2.12) {};
\node at (3.00,2.12) {};
\node at (1.24,0.35) {};
\node at (3.00,0.35) {};
\node at (3.00,0.35) {};
\node [right] at (3.44,6.53) {$\nodes_{10}$};
\node [right] at (2.56,4.77) {$\nodep_{9}$};
\node [right] at (4.32,4.77) {$\nodep_{8}$};
\node [right] at (3.44,3.00) {$\nodes_{7}$};
\node [right] at (2.56,1.24) {$\nodep_{6}$};
\node [right] at (4.32,1.24) {$\nodep_{5}$};
\end{tikzpicture}
\end{minipage}
\\[2.5ex]
\begin{minipage}{0.42\textwidth}
(D) Defining Clauses\\[1.2ex]
\begin{tabular}{llll}
\toprule
\makebox[5mm]{ID} & \multicolumn{2}{l}{Clauses} & Explanation \\
\midrule
\rtext{6} & \texttt{5 3 4} & \texttt{0} & Define $p_5$ \smallnl
\rtext{7} & \texttt{-5 -3} & \texttt{0} & \smallnl
\rtext{8} & \texttt{-5 -4} & \texttt{0} & \\
\midrule
\rtext{9} & \texttt{6 -3 -4} & \texttt{0} & Define $p_6$ \smallnl
\rtext{10} & \texttt{-6 3} & \texttt{0} & \smallnl
\rtext{11} & \texttt{-6 4} & \texttt{0} & \\
\midrule
\rtext{12} & \texttt{-7 5 6} & \texttt{0} & Define $s_7$ \smallnl
\rtext{13} & \texttt{7 -5} & \texttt{0} & \smallnl
\rtext{14} & \texttt{7 -6} & \texttt{0} & \\
\midrule
\rtext{15} & \texttt{8 1 -7} & \texttt{0} & Define $p_8$ \smallnl
\rtext{16} & \texttt{-8 -1} & \texttt{0} & \smallnl
\rtext{17} & \texttt{-8 7} & \texttt{0} & \\
\midrule
\rtext{18} & \texttt{9 -1 2 -7} & \texttt{0} & Define $p_9$ \smallnl
\rtext{19} & \texttt{-9 1} & \texttt{0} & \smallnl
\rtext{20} & \texttt{-9 -2} & \texttt{0} & \smallnl
\rtext{21} & \texttt{-9 7} & \texttt{0} & \\
\midrule
\rtext{22} & \texttt{-10 8 9} & \texttt{0} & Define $s_{10}$ \smallnl
\rtext{23} & \texttt{10 -8} & \texttt{0} & \smallnl
\rtext{24} & \texttt{10 -9} & \texttt{0} & \\
\bottomrule
\end{tabular}
\end{minipage}
\begin{minipage}{0.49\textwidth}
(E) CPOG Assertions\\[1.2ex]
\begin{tabular}{llllll}
\toprule
\makebox[5mm]{ID} & \multicolumn{2}{l}{Clause} & \multicolumn{2}{l}{Hint} & Explanation \\
\midrule
\rtext{25} & \texttt{a 5 1 3} & \texttt{0} & \rtext{3 6} & \texttt{0} & $\obar{x}_1 \land \obar{x}_3 \imply p_5$ \smallnl
\rtext{26} & \texttt{a 6 1 -3} & \texttt{0} & \rtext{4 9} & \texttt{0} & $\obar{x}_1 \land x_3 \imply p_6$ \smallnl
\rtext{27} & \texttt{a 3 7 1} & \texttt{0} & \rtext{13 25} & \texttt{0} & $\obar{x}_3 \land \obar{x}_1 \imply s_7$  \smallnl
\rtext{28} & \texttt{a 7 1} & \texttt{0} & \rtext{27 14 26} & \texttt{0} & $\obar{x}_1 \imply s_7$  \smallnl
\rtext{29} & \texttt{a 8 1} & \texttt{0} & \rtext{28 15} & \texttt{0} & $\obar{x}_1 \imply p_8$  \smallnl
\rtext{30} & \texttt{a 5 -1 3} & \texttt{0} & \rtext{1 6} & \texttt{0} & $x_1 \land \obar{x}_3 \imply p_5$ \smallnl
\rtext{31} & \texttt{a 6 -1 -3} & \texttt{0} & \rtext{2 9} & \texttt{0} & $x_1 \land x_3 \imply p_6$ \smallnl
\rtext{32} & \texttt{a 3 7 -1} & \texttt{0} & \rtext{13 30} & \texttt{0} & $\obar{x}_3 \land x_1 \imply s_7$  \smallnl
\rtext{33} & \texttt{a 7 -1} & \texttt{0} & \rtext{32 14 31} & \texttt{0} & $x_1 \imply s_7$  \smallnl
\rtext{34} & \texttt{a 9 -1} & \texttt{0} & \rtext{5 33 18} & \texttt{0} & $x_1 \imply p_9$  \smallnl
\rtext{35} & \texttt{a 1 10} & \texttt{0} & \rtext{23 29} & \texttt{0} & $\obar{x}_1 \imply s_{10}$  \smallnl
\rtext{36} & \texttt{a 10} & \texttt{0} & \rtext{35 24 34} & \texttt{0} & $s_{10}$ \\
\bottomrule
\end{tabular}
\\[1.5ex]
(F) Input Clause Deletions\\[1.2ex]
\begin{tabular}{lllll}
  \toprule
 \multicolumn{3}{l}{CPOG line} & Explanation\\
\midrule
 \texttt{d 1} & \rtext{36 8 10 12 16 21 22} & \texttt{0} & Delete clause 1 \\
 \texttt{d 2} & \rtext{36 7 11 12 16 21 22} & \texttt{0} & Delete clause 2 \\
 \texttt{d 3} & \rtext{36 8 10 12 17 19 22} & \texttt{0} & Delete clause 3 \\
 \texttt{d 4} & \rtext{36 7 11 12 17 19 22} & \texttt{0} & Delete clause 4 \\
 \texttt{d 5} & \rtext{36 16 20 22} & \texttt{0} &  Delete clause 5 \\
\bottomrule
\end{tabular}
\end{minipage}
\caption{Example formula (A), its POG representation (B), and its CPOG proof (C), (E), and (F).  The defining clauses (D) are implicitly
defined by the POG declaration (C).}
\label{fig:eg4:all}
\end{figure}

Figure \ref{fig:eg4:all} illustrates an example formula and shows how
the CPOG file declares its POG representation.  The input formula (A)
consists of five clauses over variables $x_1$, $x_2$, $x_3$, and
$x_4$.  The generated POG (B) has six nonterminal nodes representing
four products and two sums.  We name these by the node
type (product $\nodep$ or sum $\nodes$), subscripted by the ID of the
extension variable.
  The first part of the CPOG file (C) declares
these nodes using clause IDs that increment by three or four,
depending on whether the node has two children or three.  The last two
nonzero values in each sum declaration is the hint providing the
required mutual exclusion proof.

\subsection{Node Declarations}

We step through portions of the file to provide a better understanding of the CPOG proof framework.
Figure
\ref{fig:eg4:all}(D) shows the defining clauses that are implicitly
defined by the POG operation declarations.  These do not appear in the
CPOG file.  Referring back to the declarations of the sum nodes in
Figure \ref{fig:eg4:all}(C), we can see that the declaration of node
$\nodes_7$ has clause IDs 7 and 10 as the hint.  We can see in Figure
\ref{fig:eg4:all}(D) that these two clauses form a RUP proof for the clause
$\obar{p}_5 \lor \obar{p}_6$, showing that the two children of $\nodes_7$
have disjoint models.  Similarly, node $\nodes_{10}$ is declared as having
clause IDs 16 and 19 as the hint.  These form a RUP proof for the clause
$\obar{p}_8 \lor \obar{p}_9$, showing that the two children of
$\nodes_{10}$ have disjoint models.

\subsection{Forward Implication Proof}
\label{sect:forward}

Figure \ref{fig:eg4:all}(E) provides the sequence of assertions
leading to unit clause 36, consisting of the literal $s_{10}$.  This clause indicates that $\nodes_{10}$ is implied by the input clauses, i.e.,
any total assignment $\assign$
satisfying the input clauses must have its extension to $\eassign$ yield $\eassign(s_{10}) = 1$.
Working backward, we can see that
clause 35 indicates that variable $s_{10}$ will be implied by the input
clauses when $\assign(x_1) = 0$. Clause 34 indicates that node $p_9$ will
be implied by the input clauses when $\assign(x_1) = 1$, while defining clause 24 shows that node $s_{10}$ will be implied
by the input clauses when $\eassign(p_9) = 1$.  These three clauses
serve as the
hint for clause 36.

\subsection{Reverse Implication Proof}

Figure \ref{fig:eg4:all}(F) shows the RUP proof steps required to
delete the input clauses.  Consider the first of these, deleting
input clause $\obar{x}_1 \lor x_3 \lor \obar{x}_4$.  The requirement is to show
that there is no total assignment $\assign$ that falsifies this clause but extends to $\eassign$ such that $\eassign(s_{10}) = 1$.
The proof proceeds by first assuming that the clause is false, requiring
$\assign(x_1) = 1$, $\assign(x_3) = 0$, and $\assign(x_4) = 1$.  The hint then consists of unit
clauses (e.g., clause 36 asserting that $\eassign(s_{10}) = 1$) or
clauses that cause unit propagation.  Hint clauses 8 and 10 force the
assignments $\eassign(p_5) = \eassign(p_6) = 0$.  These, plus hint clause 12 force
$\eassign(s_7) = 0$.  This, plus hint clauses 16 and 21 force $\eassign(p_8) = \eassign(p_9) = 0$, leading,
via clause 22, to $\eassign(s_{10}) = 0$.  But this contradicts clause 36,
completing the RUP proof.  The deletion hints for the other input
clauses follow similar patterns---they work from the bottom nodes of
the POG upward, showing that any total assignment that falsifies the clause
must, when extended, have $\eassign(s_{10}) = 0$.

Deleting the asserted clauses is so simple that we do not show it.  It
involves simply deleting the clauses from clause number 35 down to
clause number 25, with each deletion using the same hint as was used
to add that clause.  In the end, therefore, only the defining clauses
for the POG nodes and the unit clause asserting $s_{10}$ remain,
completing a proof that the POG is logically equivalent to the input
formula.

\section{Generating CPOG from \decdnnf{}}
\label{sect:generating:cpog}

A \decdnnf{}
graph can be directly translated  into a POG.
In doing this conversion,
our program performs simplifications to
eliminate Boolean constants.
Except in degenerate cases,
where the formula is unsatisfiable or a tautology,
we can therefore assume
that the POG does not contain any constant nodes.
In addition, negation is only
applied to variables, and so the only edges with negative polarity will have variables as children.
We can therefore
view the POG as consisting
of \emph{literal} nodes corresponding to input variables and their negations, along with
\emph{nonterminal} nodes, which can be further classified as \emph{product} and \emph{sum} nodes.

\subsection{Forward Implication Proof}

For input formula $\inputformula$ and its translation into a POG $P$
with root node $\noder$, the most challenging part of the proof is to
show that $\modelset(\inputformula) \subseteq \modelset(\phi_\noder)$,
i.e., that any total assignment $\assign$ that is a model of
$\inputformula$ will extend to assignment $\eassign$ such that
$\eassign(r) = 1$, for root literal $r$.  This part of the proof
consists of a series of clause assertions leading to one adding
$\{r\}$ as a unit clause.  We have devised two methods for generating
this proof.  The \emph{monolithic} approach makes just one call to a
proof-generating SAT solver and has it determine the relationship
between the two representations.  The monolithic
approach is \emph{logically complete}, i.e., assuming the CNF formula is equivalent
to the POG, and given enough time and computing resources, it can generate a CPOG proof of equivalence.
The \emph{structural} approach only works when the POG was generated from a \decdnnf{} graph having a structure that reflects the top-down process by which it was created.
It recursively traverses the POG, generating proof obligations at each
node encountered.  It may require multiple calls to a proof-generating SAT
solver.

As notation,
let $\psi$ be a subset of the clauses in $\inputformula$.
For partial assignment
$\passign$, the expression  $\simplify{\psi}{\passign}$ denotes the set of clauses $\gamma$
obtained from $\psi$ by: (i) eliminating any
clause containing a literal $\lit$ such that $\passign(\lit) = 1$,
(ii) for the remaining clauses eliminating those literals $\lit$ for
which $\passign(\lit) = 0$, and (iii) eliminating any duplicate or tautological clauses.
In doing these simplifications, we also track the \emph{provenance}
of each simplified clause $C$, i.e., which of the (possibly multiple) input clauses simplified to become $C$.
More formally, for $C \in \simplify{\psi}{\passign}$, we let $\prov_{\passign}(C, \psi)$ denote
those clauses $C' \in \psi$, such that
$C' \subseteq C \cup \bigcup_{\lit \in \passign} \obar{\lit}$.
We then extend the definition of $\prov$ to any simplified formula
$\gamma$ as $\prov_{\passign}(\gamma, \psi) = \bigcup_{C \in \gamma} \prov_{\passign}(C, \psi)$.

The monolithic approach
takes advantage of the clausal representations of
the input formula $\inputformula$ and the POG formula $\phi_\noder$.
We can express the negation of $\phi_\noder$ in clausal form as
$\theta_{\obar{\noder}} \doteq \bigcup_{\nodeu\in P} \simplify{\theta_{u}}{\{\obar{r}\}}$.
Forward implication will hold when $\inputformula \imply \phi_\noder$, or  equivalently
when the formula $\inputformula \land \theta_{\obar{\noder}}$
is unsatisfiable, where the
conjunction can be expressed as the union
of the two sets of clauses.  The proof generator writes the clauses to a file and invokes a proof-generating SAT solver.
For each clause $C$ in the unsatisfiability proof, it adds clause $\{r\} \cup C$ to the CPOG proof, and so the empty clause in the proof becomes the unit clause $\{r\}$.
Our experimental results show
that this approach can be very effective and generates short proofs
for smaller problems, but it does not scale well enough for general
use.

The structural approach to proof generation takes the form of a recursive procedure
$\validate(\nodeu, \passign, \psi)$ taking as arguments POG
node $\nodeu$, partial assignment
$\passign$, and a set of clauses $\psi \subseteq \inputformula$.
The procedure adds a number of clauses to the proof, culminating with
the addition of the \emph{target} clause:
$u \lor \bigvee_{\lit \in \passign} \obar{\lit}$, indicating
that $(\bigwedge_{\lit \in \passign} \lit) \imply u$, i.e.,
that any total
assignment $\assign$ such that $\passign \subseteq \assign$
will extend to assignment $\eassign$ such that $\eassign(u) = 1$.
The top-level call has $\nodeu = \noder$, $\passign = \emptyset$, and $\psi = \inputformula$.
The result will therefore be to add unit clause $\{r\}$ to the proof.
Here we present a correct, but somewhat inefficient formulation of
$\validate$.  We then refine it with some optimizations.

The recursive call $\validate(\nodeu, \passign, \psi)$ assumes that we have
traversed a path from the root node down to node $\nodeu$, with the
literals encountered in the product nodes forming the partial
assignment $\passign$.  The set of clauses $\psi$ can be a proper
subset of the input clauses $\inputformula$ when a product node has caused
a splitting into clauses containing disjoint variables.
The subgraph with root node $\nodeu$ should be a POG representation of the formula
$\simplify{\psi}{\passign}$.

The process for generating such a proof depends on the form of node $\nodeu$:
\begin{enumerate}
\item If $\nodeu$ is a literal $\lit'$, then the formula
  $\simplify{\psi}{\passign}$ must consist of the single unit clause
  $C = \{\lit'\}$, such that any $C' \in \prov_{\passign}(C, \psi)$ must have $C' \subseteq \{ \lit' \} \cup\, \bigcup_{\lit \in \passign} \obar{\lit}$.
  Any of these can
  serve as the target clause.
\item If $\nodeu$ is a sum node with children $\nodeu_1$ and $\nodeu_0$,
  then, since the node originated from a \decdnnf{} graph, there must be
  some variable $x$ such that either $\nodeu_1$ is a literal node for $x$ or $\nodeu_1$ is a
  product node containing a literal node for $x$ as a child.  In either case, we
  recursively call $\validate(\nodeu_1, \passign \cup \{ x \}, \psi)$.
  This will cause the addition of the target clause
  $u_1 \lor \obar{x} \lor \bigvee_{\lit \in \passign} \obar{\lit}$.
Similarly, either $\nodeu_0$ is a literal node for $\obar{x}$ or $\nodeu_0$ is a product node containing a literal node for $\obar{x}$ as
  a child.  In either case, we recursively call $\validate(\nodeu_0, \passign \cup \{ \obar{x} \}, \psi)$,
  causing the addition of the target clause
  $u_0 \lor x \lor \bigvee_{\lit \in \passign} \obar{\lit}$.
  These recursive results can be combined with the second and third defining clauses for $\nodeu$
(see Table~\ref{tab:defining}(B))
  to generate the target clause for $\nodeu$, requiring at most two RUP steps.
\item If $\nodeu$ is a product node, then we can divide its children
  into a set of literal nodes $\lambda$ and a set of nonterminal nodes $\nodeu_1, \nodeu_2, \ldots, \nodeu_k$.
  \begin{enumerate}
    \item For each literal
  $\lit \in \lambda$, we must prove that any total assignment $\alpha$ satisfying $\psi$ and such that
  $\passign \subseteq \alpha$ has $\alpha(\lit) = 1$.  In some
  cases, this can be done by simple Boolean constraint propagation (BCP).
  In other cases, we must prove that the formula
  $\simplify{\psi}{\passign \cup \{\obar{\lit}\}}$ is unsatisfiable.  We
  do so by writing the formula to a file, invoking a proof-generating
  SAT solver, and then converting the generated unsatisfiability proof
  into a sequence of clause additions in the CPOG file.
  (The solver is constrained to only use RUP inference rules, preventing it from introducing extension variables.)
\item For a single nonterminal child ($k = 1$), we recursively call
  $\validate \left(\nodeu_1, \passign \cup \lambda, \psi\right)$.
\item For multiple nonterminal children ($k > 1$),
  it must be the case that the clauses in
  $\gamma = \simplify{\psi}{\passign \cup \lambda}$ can be partitioned into $k$ subsets
  $\gamma_1, \gamma_2, \ldots, \gamma_k$ such that $\dependencyset(\gamma_i)
  \cap \dependencyset(\gamma_j) = \emptyset$ for $1 \leq i < j \leq k$,
  and we can match each node $\nodeu_i$ to subset $\gamma_i$ based on its
  literals.
  For each $i$ such that $1 \leq i \leq k$, let $\psi_i = \prov_{\passign}(\gamma_i, \psi)$, i.e., those input clauses in $\psi$ that, when simplified, became clause partition $\gamma_i$.
  We recursively call
  $\validate \left(\nodeu_i, \passign \cup \lambda, \psi_i\right)$.
\end{enumerate}
  We then generate the target clause for node $\nodeu$ with a single RUP step,
creating the hint by combining the results from the BCP and SAT calls for
  the literals, the recursively computed target clauses, and all but
  the first defining clause for node $\nodeu$
(see Table~\ref{tab:defining}(A)).
\end{enumerate}
Observe that
all of these steps involve a polynomial number of
operations per recursive call, with the exception of those that call
a SAT solver to validate a literal.

As examples, the forward implication proof of Figure~\ref{fig:eg4:all}(E) was
generated by the structural approach.  Working from step 36 backward,
we can see that steps 35 and 36 complete the call to
$\validate(\nodes_{10}, \emptyset, \phi_I)$.  This call used $x_1$ as
the splitting variable, first calling $\validate(\nodep_{8},
\{\obar{x}_1\}, \emptyset, \phi_I)$, which completed with step 29, and
$\validate(\nodep_{9}, \{x_1\}, \phi_I)$, which completed with step
34.  We see that each of these calls required separate traversals of
nodes $\nodes_7$, $\nodep_6$, and $\nodep_5$, with the former yielding
proof steps 25--27 and the latter yielding proof steps 30--32.  This
demonstrates how our simple formulation of $\validate$ effectively
expands the graph into a tree.  This shortcoming is avoided by the use of lemmas, as is described in Section~\ref{sect:lemma}.

\subsection{Reverse Implication Proof}

Completing the equivalence proof of input formula $\inputformula$ and its POG
representation with root node $\noder$ requires showing that
$\modelset(\phi_\noder) \subseteq \modelset(\inputformula)$.  This is done in the
CPOG framework by first deleting all asserted clauses, except for the
final unit clause for root literal $r$, and then deleting all of the
input clauses.

The asserted clauses can be deleted in reverse order, using the same
hints that were used in their original assertions.  By reversing the
order, those clauses that were used in the hint when a clause was
added will still remain when it is deleted.

Each input clause deletion can be done as a single RUP step, based
on an algorithm to test for clausal entailment in \detdnnf{} graphs~\cite{darwiche:jair:2002,capelli:sat:2019}.  The
proof generator constructs the hint sequence from the defining
clauses of the POG nodes via a single, bottom-up pass through the
graph.  The RUP deletion proof for input clause $C$ effectively proves that any
total assignment $\assign$ that does not
satisfy $C$ will extend to assignment $\eassign$ such that
$\eassign(r) = 0$.  It starts with the set of literals
$\{ \obar{\lit} \mid \lit \in C\}$, describing the required condition for
assignment $\assign$ to falsify clause $C$.
It then
adds literals via unit propagation until a
conflict arises.    Unit literal $r$ gets
added right away, setting up a potential conflict.

Working upward through the graph, node $\nodeu$ is \emph{marked} when
the collected set of literals forces $\eassign(u) = 0$.  When marking $\nodeu$, the
program adds $\obar{u}$ to the RUP literals and adds the appropriate
defining clause to the hint.  A literal node for
$\lit$ will be marked if $\lit \in C$, with no hint required.  If
product node $\nodeu$ has some child $\nodeu_i$ that is marked, then
$\nodeu$ is marked and clause $i+1$ from among its defining clauses (see Table~\ref{tab:defining}(A)) is
added to the hint.  Marking sum node $\nodeu$ requires that its two children are marked.
The first defining
clause for this node (see Table~\ref{tab:defining}(B)) will then be added to the hint.  At the very end, the program
(assuming the reverse implication holds) will attempt to mark root
node $\noder$, which would require $\eassign(r) = 0$, yielding a
conflict.

It can be seen that the reverse implication proof will be polynomial in the size of the POG\@, because
each clause deletion requires a single RUP step having a hint with length
bounded by the number of POG nodes.

\section{Optimizations}
\label{sect:optimization}

The performance of the structural proof generator for forward implication, both in its execution time and
the size of the proof generated, can be improved by two optimizations
described here.  A key feature is that they do not require any changes
to the proof framework---they build on the power of extended
resolution to enable the construction of new logical structures.  They
involve declaring new product nodes to encode products of literals.
These nodes are not part of the POG
representation of the formula; they serve only to enable the forward
implication proof.

The combination of these two optimization guarantees that (i) each call
to $\validate$ for a product node will cause at most one invocation of
the SAT solver, and (ii) each call to $\validate$ for any node $\nodeu$
will cause further recursive calls only once.  Our experimental
results (Section~\ref{sect:experimental:optimize}) show that these
optimizations yield substantial benefits.

\subsection{Literal Grouping}
\label{sect:grouping}

A single recursive step of $\validate$ can encounter product nodes
having many literals as children.  The naive formulation of $\validate$
considers each literal $\lit \in \lambda$ separately.
Literal grouping allows all literals to be validated with a single call to a SAT solver.
It collects those literals
$\lit_1, \lit_2, \ldots, \lit_m$ that cannot be validated by BCP and defines a
product node $\nodev$ having these literals as children.  The goal
then becomes to prove that any total assignment $\assign$ consistent with the partial assignment $\passign$,
must, when extended to $\eassign$, yield $\eassign(v) = 1$.
A single call to the solver can generate this proof by invoking it on the formula
  $\simplify{\psi}{\passign} \cup \simplify{\theta_{v}}{\{ \obar{v} \}}$, which should be unsatisfiable.
  The proof steps can be mapped back into clause addition steps in the CPOG file, incorporating the
  input clauses and the defining clauses for $\nodev$ into the hints.

\subsection{Lemmas}
\label{sect:lemma}

As we have noted, the recursive calling of $\validate$ starting at
root $\noder$ effectively expands the POG into a tree, and this can
lead to an exponential number of calls.
These shared subgraphs arise when the knowledge compiler employs \emph{clause caching}
to detect that the simplified set of
clauses arising from one partial assignment to the literals matches that
of a previous partial assignment~\cite{darwiche:aaai:2002}.
When this \decdnnf{} node is translated into POG
node $\nodeu$, the proof generator can assume (and also check), that
there is a simplified set of clauses $\gamma_{\nodeu}$
for which the subgraph with root $\nodeu$ is its POG representation.

The proof generator can exploit the sharing of subgraphs
by constructing and proving a \emph{lemma} for each node
$\nodeu$ having $\indegree(\nodeu) > 1$.  This proof shows that any
total assignment $\assign$ that satisfies formula $\gamma_{\nodeu}$ must extend to assignment $\eassign$ such that
$\eassign(u) = 1$.  This lemma is then invoked for every node having
$\nodeu$ as a child.
As a result, the generator will make recursive calls during a call to $\validate$ only once for each node in the POG\@.

The challenge for implementing this strategy is to find a way to
represent the clauses for the simplified formula $\gamma_{\nodeu}$ in the CPOG file.  Some may be
unaltered input clauses, and these can be used directly.  Others,
however will be clauses that do not appear in the input formula.  We
implement these by adding POG product nodes to the CPOG file to create
the appropriate clauses.  Consider an \emph{argument} clause
$C \in \gamma_{\nodeu}$ with $C = \lit_1 \lor \lit_2 \lor \cdots \lor \lit_k$.  If we
define a product node $\nodev$ with arguments
$\obar{\lit}_1, \obar{\lit}_2, \ldots, \obar{\lit}_k$, 
we will introduce a defining clause
$v \lor \lit_1 \lor \lit_2 \lor \cdots \lit_k$.  We call this a {\em
  synthetic} clause having $\obar{v}$ as the \emph{guard literal}.
That is, a partial assignment $\passign$ such that $\passign(v) = 0$ will {\em
  activate} the clause, causing it to represent argument clause $C$.  On the other
hand, a partial assignment with $\passign(v) = 1$ will
cause the clause to become a tautology and therefore have no effect.

Suppose for every clause $C_j \in \gamma_{\nodeu}$ that does not correspond to
an input clause, we generate a synthetic clause $C'_j$ with guard literal
$\obar{v}_j$, for $1 \leq j \leq m$.  Let $\gamma'_{\nodeu}$ be the formula where each clause $C_j$ is replaced by synthetic clause $C'_j$,
while input clauses in $\gamma_{\nodeu}$ are left unchanged.
Let $\lassign = \{ \obar{v}_1, \obar{v}_2, \ldots, \obar{v}_m \}$.
Invoking $\validate(\nodeu, \lassign, \gamma'_{\nodeu})$
 will then prove a lemma, given by the target clause
 $u \lor v_1 \lor v_2 \lor \cdots \lor v_m$,
 showing that any total assignment $\assign$ that activates the synthetic clauses will cause $u$ to be assigned $1$.
 More precisely, given assignment $\assign$ and its extension $\eassign$, if $\eassign(v_j) = 0$ for every guard literal $\obar{v}_j$, then $\eassign(u) = 1$.

Later, when node $\nodeu$ is encountered by a call to $\validate(\nodeu, \passign, \psi)$, we invoke the lemma
by showing that each synthetic clause
$C_j$ matches some simplified clause in $\simplify{\psi}{\passign}$.  More precisely,
for $1 \leq j \leq m$,
we use clause addition to assert the clause
$\obar{v}_j \lor \bigvee_{\lit \in \passign} \obar{\lit}$,
showing that synthetic clause $C_j$ will be activated.
Combining the lemma with these activations provides a derivation of the target clause for the call to $\validate$.

Observe that the lemma structure can be hierarchical, since a shared
subgraph may contain nodes that are themselves roots of shared
subgraphs.  Even then, the principles described allow the
definition, proof, and applications of a lemma for each shared node in
the graph.  For any node $\nodeu$, the first call to
$\validate(\nodeu, \passign, \psi)$ may require further recursion,
but any subsequent call can simply reuse the lemma proved by the first call.

\subsection{Lemma Example}
\label{app:lemma:eg}

\begin{figure}
(A) Additional nodes\\[1.0em]
\begin{tabular}{llll}
\toprule
\makebox[5mm]{ID} & \multicolumn{2}{l}{CPOG line} & Explanation \\
\midrule
\rtext{25} & \texttt{p 11 -3 4} & \texttt{0} & $v_{11} = \obar{x}_3 \pand {x}_4$ \\
\rtext{28} & \texttt{p 12  3 -4} & \texttt{0} & $v_{12} = {x}_3 \pand \obar{x}_4$ \\
\bottomrule
\end{tabular}
\\[1.0em]
(B) Implicit Clauses\\[1.2em]
\begin{tabular}{llll}
\toprule
\makebox[5mm]{ID} & \multicolumn{2}{l}{Clauses} & Explanation \\
\midrule
\rtext{25} & \texttt{11 3 -4} & \texttt{0} & Argument clause $\{x_3 ,\, \obar{x}_4\}$, activated by $\obar{v}_{11}$ \\
\rtext{26} & \texttt{-11 -3} & \texttt{0} & \\
\rtext{27} & \texttt{-11 4} & \texttt{0} & \\
\midrule
\rtext{28} & \texttt{12 -3 4} & \texttt{0} & Argument clause $\{\obar{x}_3,\,  {x}_4\}$, activated by $\obar{v}_{12}$ \\
\rtext{29} & \texttt{-12 3} & \texttt{0} & \\
\rtext{30} & \texttt{-12 -4} & \texttt{0} & \\
\bottomrule
\end{tabular}
\\[1.0em]
(C) CPOG Assertions\\[1.0em]
\begin{tabular}{llllll}
\toprule
\makebox[5mm]{ID} & \multicolumn{2}{l}{Clause} & \multicolumn{2}{l}{Hint} & Explanation \\
\midrule
\multicolumn{6}{l}{Lemma Proof} \\
\rtext{31} & \texttt{a 5 11 12 3} & \texttt{0} & \rtext{25 6} & \texttt{0} & $(\obar{v}_{11} \land \obar{v}_{12}) \land \obar{x}_3 \imply p_{5}$ \\
\rtext{32} & \texttt{a 6 11 12 -3} & \texttt{0} & \rtext{28 9} & \texttt{0} & $(\obar{v}_{11} \land \obar{v}_{12}) \land {x}_3 \imply p_{6}$ \\
\rtext{33} & \texttt{a 3 7 11 12} & \texttt{0} & \rtext{13 31} & \texttt{0} & $(\obar{v}_{11} \land \obar{v}_{12}) \land \obar{x}_3 \imply s_{7}$ \\
\rtext{34} & \texttt{a 7 11 12} & \texttt{0} & \rtext{33 14 32} & \texttt{0} & $(\obar{v}_{11} \land \obar{v}_{12}) \imply s_{7}$ \\
\midrule
\multicolumn{6}{l}{Lemma Application \#1} \\
\rtext{35} & \texttt{a -11 1} & \texttt{0} & \rtext{26 27 3} & \texttt{0} & $\obar{x}_1 \imply \obar{v}_{11}$ \\
\rtext{36} & \texttt{a -12 1} & \texttt{0} & \rtext{29 30 4} & \texttt{0} & $\obar{x}_1 \imply \obar{v}_{12}$ \\
\rtext{37} & \texttt{a 7 1} & \texttt{0} & \rtext{35 36 34} & \texttt{0} & $\obar{x}_1 \imply s_7$ \\
\midrule
\rtext{38} & \texttt{a 8 1} & \texttt{0} & \rtext{37 15} & \texttt{0} & $\obar{x}_1 \imply p_8$ \\
\midrule
\multicolumn{6}{l}{Lemma Application \#2} \\
\rtext{39} & \texttt{a -11 -1} & \texttt{0} & \rtext{26 27 1} & \texttt{0} & ${x}_1 \imply \obar{v}_{11}$ \\
\rtext{40} & \texttt{a -12 -1} & \texttt{0} & \rtext{29 30 2} & \texttt{0} & ${x}_1 \imply \obar{v}_{12}$ \\
\rtext{41} & \texttt{a 7 -1} & \texttt{0} & \rtext{39 40 34} & \texttt{0} & ${x}_1 \imply s_7$ \\
\midrule
\rtext{42} & \texttt{a 9 -1} & \texttt{0} & \rtext{5 41 18} & \texttt{0} & ${x}_1 \imply p_9$ \\
\rtext{43} & \texttt{a 1 10} & \texttt{0} & \rtext{23 38} & \texttt{0} & $\obar{x}_1 \imply s_{10}$  \\
\rtext{44} & \texttt{a 10} & \texttt{0} & \rtext{43 24 42} & \texttt{0} & $s_{10}$ \\
\bottomrule
\end{tabular}
\caption{Example of lemma definition, proof, and application}
\label{fig:eg4:lemmas}
\end{figure}

Figure~\ref{fig:eg4:lemmas} shows an alternate forward implication
proof for the example of Figure~\ref{fig:eg4:all} using a lemma to
represent the shared node $\nodes_7$.  We can see that the POG with
this node as root encodes the Boolean formula $x_3 \leftrightarrow x_4$, having a CNF representation consisting of the clauses
$\{x_3 ,\, \obar{x}_4\}$ and $\{\obar{x}_3 ,\, {x}_4\}$.  The product node
declarations shown in Figure~\ref{fig:eg4:lemmas}(A) create synthetic
clauses 25 and 28 to encode these arguments with activating literals
$\obar{v}_{11}$ and $\obar{v}_{12}$, respectively.  Clauses 31--34
then provide a proof of the lemma, stating that any assignment
$\assign$ that activates these clauses will, when extended,  assign $1$ to $s_7$.
Clauses 35 and 36 state that an assignment with $\assign(x_1) = 0$
will, when extended, cause the first synthetic clause to activate due to input clause
3, and it will cause the second synthetic clause to activate due to
input clause 4.  From this, clause 37 can use the lemma to state that
assigning $0$ to $x_1$ will cause $s_7$ to evaluate to $1$.  Similarly,
clauses 39 and 40 serve to activate the synthetic clauses when
$\assign(x_1) = 1$, due to input clauses 1 and 2, and clause 41 then
uses the lemma to state that assigning $1$ to $x_1$ will cause $s_7$ to
evaluate to $1$.

In this example, adding the lemma increases the proof length, but that
is only because it is such a simple formula.



\section{A Formally Verified Toolchain}
\label{sect:formally-verified-toolchain}
\label{sect:lean:subtle-condition}

We set out to formally verify the system with two goals in mind:
first, to ensure that the CPOG framework is mathematically sound;
and second, to implement correct-by-construction proof checking
and ring evaluation (the ``Trusted Code'' components of Figure~\ref{fig:chain}).
These two goals are achieved with a single proof development
in the \lean{} programming language~\cite{demoura:cade:2021}.
Verification was greatly aided
by the Aesop~\cite{23limperg_aesop_white_box_best_first_proof_search_lean}
automated proof search tactic.
\lean{} is based on a logical foundation
in which expressions have a computational interpretation.
As in other proof assistants such as Isabelle~\cite{nipkow:et:al:02} and Coq~\cite{coq},
functions defined in the formal system
can be compiled to machine code.
At the same time,
we can state and prove claims about them within the same system,
thereby verifying that our functions compute the intended results.
In this section,
we describe the functionality we implemented,
what we proved about it,
and the assumptions we made.

\paragraph{Data structures and mathematical model.}
When thinking about formal verification,
it is helpful to distinguish between data structures
that play a role in the code being executed,
and \emph{ghost} definitions that serve as a mathematical model,
allowing us to state and prove specifications,
but are erased during compilation and not executed.
In the codebase,
we generally store definitions in the two classes
under {\tt Data/} and {\tt Model/},
respectively.

Among the former is our representation of CNF formulas.
Following the DIMACS CNF convention,
a variable is represented as a positive natural number,
a literal is a non-zero integer,
a clause is an array of literals,
and a CNF formula is an array of clauses.
\begin{lstlisting}
def Var := { x : Nat // 0 < x }
def ILit := { i : Int // i ≠ 0 }
abbrev IClause := Array ILit
abbrev ICnf := Array IClause
\end{lstlisting}

A POG is represented as a flat array of elements.
Each element {\tt PogElt} of a POG is either a variable,
a binary disjunction (sum),
or an arbitrary conjunction (product).
\begin{lstlisting}
inductive PogElt where
  | var (x : Var) : PogElt
  | disj (x : Var) (l r : ILit) : PogElt
  | conj (x : Var) (args : Array ILit) : PogElt
\end{lstlisting}
In the first case,
the argument \lstinline{x} is the index of an input variable;
in disjunctions and conjunctions,
it is an extension variable appearing in the CPOG file.
A \lstinline{Pog} is then an array of \lstinline{PogElt}s that is well-founded
in the sense that each element depends only on prior elements in the array.
Note that representing edges as literals
allows us to negate the arguments to \lstinline{disj} and \lstinline{conj}.

On the mathematical side,
our specifications rely on a general theory of propositional logic
mirroring Section~\ref{sect:logical:foundations}.
The type \lstinline{PropForm} describes the syntax of propositional formulas.
It is generic over the type of variables,
so we instantiate it with numeric variables as \lstinline{PropForm Var}.
\begin{lstlisting}
inductive PropForm (ν : Type u)
  | var (x : ν)
  | tr
  | fls
  | neg    (φ : PropForm ν)
  | conj   (φ₁ φ₂ : PropForm ν)
  | disj   (φ₁ φ₂ : PropForm ν)
  | impl   (φ₁ φ₂ : PropForm ν)
  | biImpl (φ₁ φ₂ : PropForm ν)
\end{lstlisting}
Assignments of truth values are taken to be total functions \lstinline{PropAssignment Var := Var → Bool}.
Requiring totality is not a limitation:
instead of talking about two equal,
partial assignments to a subset $X' \subseteq X$ of variables,
we can more conveniently talk about two total assignments that agree on $X'$.
We write \lstinline{σ ⊨ φ} when \lstinline{σ : PropAssignment Var} satisfies \lstinline{φ : PropForm Var}.

Functions \lstinline{ILit.toPropForm},
\lstinline{IClause.toPropForm},
\lstinline{ICnf.toPropForm},
and \lstinline{Pog.toPropForm}
relate data structures to the formulas they encode.
For example, given a literal \lstinline{u},
\lstinline{P.toPropForm u} denotes the interpretation
of the node $\nodeu$ corresponding to \lstinline{u} in the POG \lstinline{P}
as a propositional formula $\phi_\nodeu/\neg\phi_\nodeu$
over the input variables.
It is negated if \lstinline{u} has negative polarity.
Lean provides a convenient ``anonymous projection'' notation
that allows writing \lstinline{P.toPropForm u} instead of \lstinline{Pog.toPropForm P u}
when \lstinline{P} has type \lstinline{Pog},
\lstinline{C.toPropForm} instead of \lstinline{IClause.toPropForm C}
when \lstinline{C} has type \lstinline{IClause},
etc.

In order to reason about composite formulas,
we found it easier to work with propositional formulas modulo logical equivalence,
a structure known in logic as the \emph{Lindenbaum--Tarski algebra},
rather than using \lstinline{PropForm} directly.
Its advantage is that equivalent but not syntactically equal formulas
(such as $x \vee \neg x$ and $\top$)
give rise to equal elements in the algebra,
and equality has a privileged position in proof assistants based on type theory:
equals can be substituted for equals in any context.
In this way, forgetting syntactic detail is helpful.
On the other hand,
using the algebra gives rise to some challenges.
The algebra, called {\tt PropFun}, is defined as a quotient,
with Boolean operations and the entailment relation
lifted from the syntax of formulas to the new type.
It is no longer straightforward to say
when an element of the quotient ``depends'' on a variable
since equivalent formulas can refer to different sets of variables.
Instead, we use a semantic notion of dependence
in which an element $\phi$ of the quotient depends on a variable $x$
if and only if there is a truth assignment that satisfies $\phi$,
but falsifies $\phi$ after $x$ is flipped.
\begin{lstlisting}
/-- The semantic variables of `φ` are those it is sensitive to as a Boolean
function. Unlike `vars`, this set is stable under equivalence of formulas. -/
def semVars (φ : PropFun ν) : Set ν :=
  { x | ∃ (τ : PropAssignment ν), τ ⊨ φ ∧ τ.set x (!τ x) ⊭ φ }
\end{lstlisting}

\paragraph{Proof checking.}
The goal of a CPOG proof is to construct a POG
that is equivalent to the input CNF $\inputformula$.
The database of active clauses,
the POG being constructed,
and its root literal,
are stored in a checker state structure {\tt PreState}.
The checker begins by parsing the input formula,
initializing the active clauses to $\theta \leftarrow \inputformula$,
and initializing the POG $P$ to an empty one.
It then processes every step of the CPOG proof,
either modifying its state by adding/deleting clauses in $\theta$
and adding nodes to $P$,
or throwing an exception if a step is incorrect.
Afterwards, it carries out the \textsc{final conditions} check of Section~\ref{subsection:semantics}.

Throughout the process,
we maintain invariants needed to establish the final result.
These ensure that $P$ is partitioned
and that a successful final check entails the logical equivalence
of $\inputformula$ and $\phi_\noder$,
where $\noder$ is the final POG root (Theorem~\ref{thm:lean:equiv}).
Formally, we define a type {\tt State}
consisting of those {\tt PreState}s
that satisfy all the invariants.
A {\tt State} is a structure
combining {\tt PreState} fields with additional ones
storing computationally irrelevant \emph{ghost state}
that asserts the invariants.
The fields of \lstinline{st : PreState} include
\lstinline{st.inputCnf} for $\inputformula$,
\lstinline{st.clauseDb} for $\theta$,
and \lstinline{st.pog} for $P$.
We write \lstinline{st.pogDefsForm} for the clausal POG definitions formula $\bigwedge_{\nodeu\in P}\theta_u$,
and \lstinline{st.allVars} for all variables (original and extension) added so far.
For any $\nodeu\in P$,
\lstinline{st.pog.toPropForm u} computes $\phi_\nodeu$.

The first invariant states that assignments to input variables
extend uniquely to extension variables defining the POG nodes.
In the formalization, we split this into extension and uniqueness:

\begin{minipage}{\textwidth}
\begin{lstlisting}
/-- Any assignment satisfying φ₁ extends to φ₂ while preserving values on X. -/
def extendsOver (X : Set Var) (φ₁ φ₂ : PropForm Var) :=
 ∀ (σ₁ : PropAssignment Var), σ₁ ⊨ φ₁ → ∃ σ₂, σ₁.agreeOn X σ₂ ∧ σ₂ ⊨ φ₂
/-- Assignments satisfying φ are determined on Y by their values on X. -/
def uniqueExt (X Y : Set Var) (φ : PropForm Var) :=
 ∀ (σ₁ σ₂ : PropAssignment Var), σ₁ ⊨ φ → σ₂ ⊨ φ → σ₁.agreeOn X σ₂ →
    σ₁.agreeOn Y σ₂

invariants.extends_pogDefsForm : extendsOver st.inputCnf.vars ⊤ st.pogDefsForm
invariants.uep_pogDefsForm : uniqueExt st.inputCnf.vars st.allVars st.pogDefsForm
\end{lstlisting}
\end{minipage}
Note that in the definition of \lstinline{uniqueExt}, the arrows associate to the right,
so the definition says that the three assumptions imply the conclusion.
The next invariant guarantees that the set of solutions over the input variables is preserved:
\begin{lstlisting}
def equivalentOver (X : Set Var) (φ₁ φ₂ : PropForm Var) :=
  extendsOver X φ₁ φ₂ ∧ extendsOver X φ₂ φ₁

invariants.equivInput : equivalentOver st.inputCnf.vars st.inputCnf st.clauseDb
\end{lstlisting}
Finally, for every node $\nodeu\in P$ with corresponding literal $u$ we ensure that $\phi_\nodeu$ is partitioned (Definition~\ref{def:partitioned-operation-formula}) and relate $\phi_\nodeu$ to its clausal encoding $\theta_u \doteq u \wedge \bigwedge_{\nodev\in P}\theta_v$:
\begin{lstlisting}
def partitioned : PropForm Var → Prop
  | tr | fls | var _ => True
  | neg φ    => φ.partitioned
  | disj φ ψ => φ.partitioned ∧ ψ.partitioned ∧ ∀ τ, ¬(τ ⊨ φ ∧ τ ⊨ ψ)
  | conj φ ψ => φ.partitioned ∧ ψ.partitioned ∧ φ.vars ∩ ψ.vars = ∅
\end{lstlisting}
\begin{lstlisting}
invariants.partitioned : ∀ (u : ILit), (st.pog.toPropForm u).partitioned
invariants.equivalent_lits : ∀ (u : ILit), equivalentOver st.inputCnf.vars
    (u ∧ st.pogDefsForm) (st.pog.toPropForm x)
\end{lstlisting}

The bulk of our work involved showing
that these invariants are indeed maintained by the checker
when going through a valid CPOG proof,
modifying the active clause database and the POG.
Together with additional, technical invariants
about the correctness of cached computations,
they imply the soundness theorem for $P$ with root node $\noder$:

\begin{thm}
\label{thm:lean:equiv}
If the proof checker has assembled POG $P$ with root node $\noder$ starting from input formula $\inputformula$, and \textsc{final conditions} (as stated in Section~\ref{subsection:semantics}) hold of the checker state, then $\inputformula$ is logically equivalent to $\phi_\noder$.
\end{thm}
\begin{proof}
\textsc{Final conditions} imply that the active clausal formula $\theta$ is exactly $\pogformula \doteq \{\{r\}\} \cup \; \bigcup_{\nodeu \in P} \theta_{u}$. The conclusion follows from this and the checker invariants. The full proof is formally verified in Lean.
\end{proof}
After certifying a CPOG proof, the checker can pass its in-memory POG representation to the ring evaluator, along with the partitioning guarantee provided by \texttt{invariants.partitioned}.

\paragraph{Ring evaluation.} We formalized the central quantity (\ref{eqn:rep}) in the ring evaluation problem
(Definition \ref{def:ring_evaluation}) in a commutative ring \lstinline{R} as follows:
\begin{lstlisting}
def weightSum {R : Type} [CommRing R]
    (weight : Var → R) (φ : PropForm Var) (s : Finset Var) : R :=
  ∑ τ in models φ s, ∏ x in s, if τ x then weight x else 1 - weight x
\end{lstlisting}
The rules for efficient ring evaluation of partitioned formulas are expressed as:
\begin{lstlisting}
def ringEval (weight : Var → R) : PropForm Var → R
  | tr       => 1
  | fls      => 0
  | var x    => weight x
  | neg φ    => 1 - ringEval weight φ
  | disj φ ψ => ringEval weight φ + ringEval weight ψ
  | conj φ ψ => ringEval weight φ * ringEval weight ψ
\end{lstlisting}
Proposition~\ref{prop:ring:eval} is then formalized as follows:
\begin{lstlisting}
theorem ringEval_eq_weightSum (weight : Var → R) {φ : PropForm Var} :
    partitioned φ → ringEval weight φ = weightSum weight φ (vars φ)
\end{lstlisting}
To efficiently compute the ring evaluation of a formula represented by a POG node, we implemented
\lstinline{Pog.ringEval} and then proved that it matches the specification above:
\begin{lstlisting}
theorem ringEval_eq_ringEval (pog : Pog) (weight : Var → R) (x : Var) :
  pog.ringEval weight x = (pog.toPropForm x).ringEval weight
\end{lstlisting}
Applying this to the output of our verified CPOG proof checker, which we know to be partitioned
and equivalent to the input formula $\inputformula$, we obtain a proof that our toolchain computes
the correct ring evaluation of $\inputformula$.

\paragraph{Model counting.} Finally, we established that ring evaluation with the appropriate weights
corresponds to the standard model count. To do so, we defined a function that
carries out an integer calculation of the number of models over a set of variables
of cardinality \lstinline{nVars}:
\begin{lstlisting}
def countModels (nVars : Nat) : PropForm Var → Nat
  | tr       => 2^nVars
  | fls      => 0
  | var _    => 2^(nVars - 1)
  | neg φ    => 2^nVars - countModels nVars φ
  | disj φ ψ => countModels nVars φ + countModels nVars ψ
  | conj φ ψ => countModels nVars φ * countModels nVars ψ / 2^nVars
\end{lstlisting}
We then formally proved that for a partitioned formula whose variables are among a finite set
\lstinline{s}, this computation really does count the number of models over \lstinline{s}:
\begin{lstlisting}
theorem countModels_eq_card_models {φ : PropForm Var} {s : Finset Var} :
  vars φ ⊆ s → partitioned φ → countModels (card s) φ = card (models φ s)
\end{lstlisting}
In particular, taking \lstinline{s} to be exactly the variables appearing in \lstinline{φ},
we have that the number of models of \lstinline{φ} over its variables is
\lstinline{countModels φ (card (vars φ))}.

\paragraph{Trust.}
To conclude this section, let us clarify what has been verified and what has to be trusted.
Recall that our first step is to parse CNF and CPOG files in order to read in the initial
formula and the proof. We do not verify this step. Instead, the verified checker exposes
flags \verb+--print-cnf+ and \verb+--print-cpog+ which reprint the consumed formula or proof,
respectively. Comparing this to the actual files using {\tt diff} provides an easy way of ensuring
that what was parsed matches their contents. This involves trusting only the correctness of the
print procedure and {\tt diff}. Similarly, if one wants to establish the correctness of the POG
contained in the CPOG file, one can print out the POG that is constructed by the checker and compare.

Lean's code extraction replaces calculations on natural numbers and integers with
efficient but unverified arbitrary precision versions.
Lean also uses an efficient implementation of arrays; within the
formal system, these are defined in terms of lists, but code extraction replaces them
with dynamic arrays and uses reference counting to allow destructive updates when it is safe
to do so \cite{Ullrich:de:Moura:19}.

In summary, in addition to trusting Lean's foundation and kernel checker,
we also have to trust that code extraction respects that foundation,
that the implementations of basic data structures satisfy their descriptions,
and that, after parsing, the computation uses the correct input formula.
All of our specifications have been completely proven and verified relative to these assumptions.



\section{Implementation}
We have implemented programs that, along with
the \dfour{} knowledge compiler, form the toolchain illustrated in
Figure~\ref{fig:chain}.\footnote{The source code for all tools, as well as the \lean{} derivation and checker,
is available at \url{https://github.com/rebryant/cpog/releases/tag/v1.0.0}.  Upon acceptance of this paper, we will create an archival version of the code, as well as the experimental results, on Zenodo.}  The proof generator is the same in both
cases, since it need not be trusted.
Our \emph{verified}
version of the proof checker and ring evaluator have been formally
verified within the \lean{} theorem prover.  Our long term goal is to
rely on these.  Our \emph{prototype} version is written in C\@.
It is faster and
more scalable, but we anticipate its need will diminish as the
verified version is further optimized.

Our proof generator is written in C/C++ and uses a recent version of
the \cadical{} SAT solver that directly generates hinted proofs in
LRAT format~\cite{biere:sat:2023}.  It also uses their tool \ltrim{}
to reduce the length of the generated proofs.

Section~\ref{sect:forward} presented two methods for generating the
forward implication proof: a monolithic method relying on a single
call to a proof-generating SAT solver, and a structural method that
traverses the POG recursively and generates proof assertions for each
node encountered.  We devised an approach that combines the two,
forming our \emph{hybrid} method.  Based on problem parameters, this approach
starts with a top-down recursion, as with the structural method, but
it shifts to a monolithic method once the subgraph size drops below a
threshold.  Section~\ref{sect:experimental:hybrid} describes the
experiments used to determine the parameters for this approach in more
detail.

The proof generator can optionally be instructed to generate a {\em
one-sided} proof, providing only the reverse-implication portion of the proof via
input clause deletion.  This can provide useful information---any
assignment that is a model for the compiled representation
must also be a model for the input formula---even when
full validation is impractical.

We incorporated a ring evaluator into the prototype checker.  It can
perform both unweighted and weighted model counting with full precision.
It performs arithmetic over a subset of the rationals we call
$\drational$, consisting of numbers of the form $a \cdot 2^{b} \cdot
5^{c}$, for integers $a$, $b$, and $c$, and with $a$ implemented to
have arbitrary range.  Allowing scaling by powers of 2 enables the
density computation and rescaling required for unweighted model
counting.  Allowing scaling by powers of both 2 and 5 enables exact
decimal arithmetic, handling the weights used in the weighted model
counting competitions.  To give a sense of scale, the counter
generated a result with 260,909 decimal digits for one of the weighted
benchmarks.  Our implementation of arbitrary-range integers represents
a number as a sequence of ``digits'' with each digit ranging from $0$
to $10^9-1$, and with the digits stored as four-byte blocks.  This
allows easy conversion to and from a decimal representation of the
number.

\section{Experimental Evaluation}
\label{sect:experimental}

Our experimental results seek to answer the following questions:
\begin{itemize}
\item How can a hybrid approach for the forward implication proof generation take advantage of the relative strengths of the monolithic and structural approaches?
\item How well does our toolchain perform on actual benchmark problems?
\item How strongly does our toolchain rely on the structure of the POG?
\item How effective are the optimizations presented in Section~\ref{sect:optimization}?
\item How does the verified proof checker perform, relative to the prototype checker?
\item How does our toolchain perform compared to other tools for verifying the results of knowledge compilation and model counting?
\end{itemize}

\subsection{Methodology}
\label{sect:benchmark}

All experiments were run on a 2021 Apple MacBook Pro, with a 3.2~Ghz
Apple M1 processor and 64~GB of RAM\@.  We used a Samsung T7
solid-state disk (SDD) for file storage.  We found that using an SSD was
critical for dealing with the very large proof files (some over 150~GB).

As described in Section~\ref{sect:pog}, we define the size of POG $P$
to be to be the the number of nonterminal nodes plus the number of
edges from these nodes to their children.  This is also equal to the
total number of defining clauses for the POG sum and product
operations.

For benchmark problems, we used the public problems from the
2022 unweighted and weighted model counting competitions.\footnote{Downloaded
from
\url{https://mccompetition.org/2022/mc\_description.html}}
We
found that there were 180 unique CNF files among these, ranging in
size from 250 to 2,753,207 clauses.  With a runtime limit of 4,000
seconds, \dfour{} completed for 123 of the benchmark problems.  Our
proof generator was able to convert all but one of these into POGs,
with their declarations ranging from 304 to 1,765,743,261 (median
774,883) defining clauses.  The additional problem would require
2,761,457,765 defining clauses, and this count overflowed the 32-bit
signed integer we use to represent clause identifiers.

To make some of the experiments more tractable, we also created a
reduced benchmark set, consisting of 90 out of the 123 problems
for which \dfour{} ran in at most 1000 seconds, and the generated POG
had at most $10^7$ defining clauses.  These ranged in size from 304 to
8,493,275 defining clauses, with a median of 378,325.

Over the course of our tool development and evaluation, we have run \dfour{} thousands of times.
Significantly, we have not encountered any case where \dfour{} generated an incorrect result.

We found that computing the \emph{tree ratio} of a POG provides a
useful metric for the degree of sharing among subgraphs.
Formally, define the \emph{tree size} of node $\nodeu$, denoted $\treesize(\nodeu)$, recursively:
\begin{itemize}
\item When $\nodeu$ is a terminal node, $\treesize(\nodeu) = 0$.
\item  When $\nodeu$ is a nonterminal node, with children $\nodeu_1, \nodeu_2, \ldots, \nodeu_k$:
  \begin{eqnarray}
\treesize(\nodeu) & =& k + 1 + \sum_{1 \leq i \leq k} \treesize(\nodeu_i) \label{eqn:treesize}
  \end{eqnarray}
\end{itemize}
A POG $P$ with root node $\noder$ is then defined to have a tree ratio $\treesize(\noder)/|P|$.
The tree size of a POG measures its size if all shared subgraph were expanded such that the graph is transformed into a tree.
The tree ratio then measures the extent of subgraph sharing.
The 122 problems for which POGs were generated
had tree ratios ranging between $1.0$ and $52{,}410$, with a median of
$11.6$.  Considering that the tree size can be exponentially larger
than the POG size, these ratios are fairly modest.

\subsection{Designing a Hybrid Forward-Implication Proof Generator}
\label{sect:experimental:hybrid}

\begin{figure}
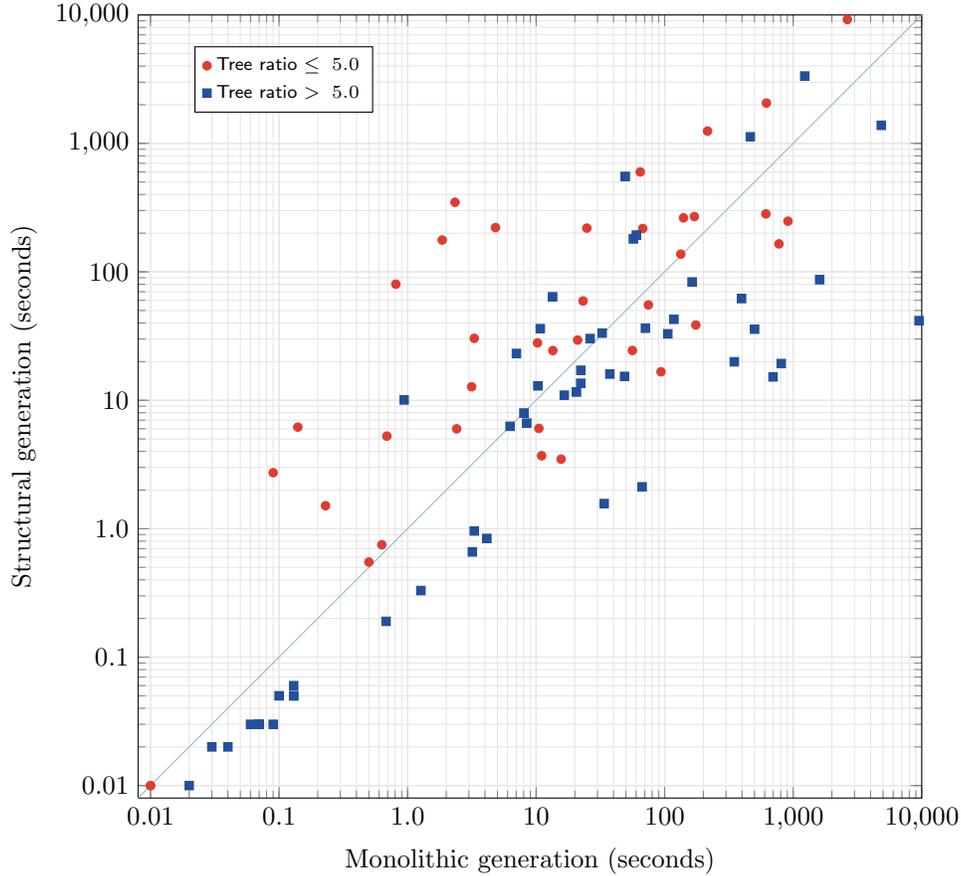

\centering{%
\begin{tikzpicture}
  \begin{axis}[mark options={scale=0.55},height=12cm,width=12cm,grid=both, grid style={black!10}, ymode=log,
      legend style={at={(0.30,0.96)}},
      legend cell align={left},
                              xmode=log,xmin=0.008,xmax=10000,
                              xtick={0.01, 0.1,1.0,10,100,1000,10000}, xticklabels={0.01, 0.1, 1.0, 10, 100, {1,000}, {10,000}},
                              ymode=log, ymin=0.008, ymax=10000,
                              ytick={0.01, 0.1,1.0,10,100,1000,10000}, yticklabels={0.01, 0.1, 1.0, 10, 100, {1,000}, {10,000}},
                              xlabel={Monolithic generation (seconds)}, ylabel={Structural generation (seconds)},
            ]

    \input{data-formatted/sub10M-seconds-subthresh}
    \input{data-formatted/sub10M-seconds-supthresh}
    \legend{
      \scriptsize \textsf{Tree ratio $\leq \; 5{.}0$},
      \scriptsize \textsf{Tree ratio $> \; 5{.}0$},
    }
    \addplot[mark=none, color=lightblue] coordinates{(0.001,0.001) (10000,10000)};

          \end{axis}
\end{tikzpicture}
} 
\caption{Structural (Y axis) versus monolithic (X axis) forward implication proof generation times.  The structural approach generally performed better
for formulas with high tree ratios.}
\label{fig:seconds:mono:structural}
\end{figure}

\begin{figure}
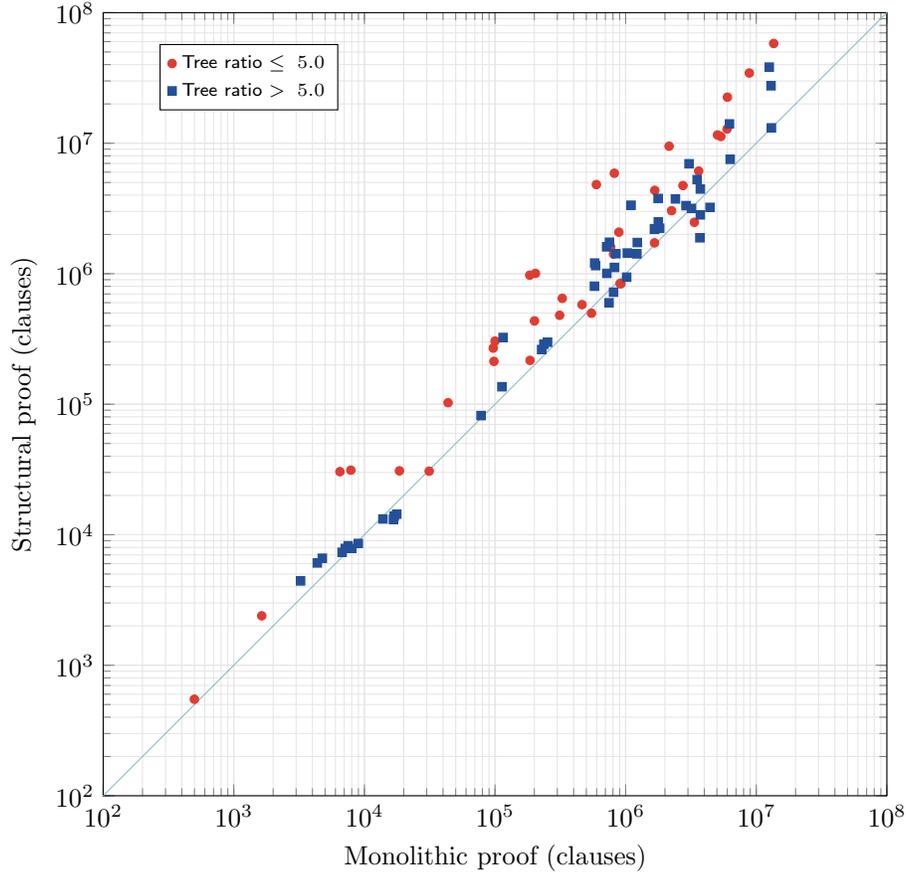

\centering{%
\begin{tikzpicture}
  \begin{axis}[mark options={scale=0.55},grid=both, height=12cm,width=12cm, grid style={black!10}, ymode=log,
      legend style={at={(0.30,0.96)}},
      legend cell align={left},
                              xmode=log,xmin=1e2,xmax=1e8,
                              xtick={100,1000, 10000, 100000, 1000000, 10000000, 100000000, 1000000000, 1e10}, 
                              xticklabels={$10^2$,$10^3$,$10^4$,$10^5$,$10^6$,$10^7$,$10^8$,$10^9$, $10^{10}$},
                              ymode=log,ymin=1e2,ymax=1e8,
                              ytick={100,1000, 10000, 100000, 1000000, 10000000, 100000000, 1000000000, 1e10}, 
                              yticklabels={$10^2$,$10^3$,$10^4$,$10^5$,$10^6$,$10^7$,$10^8$,$10^9$, $10^{10}$},
                              xlabel={Monolithic proof (clauses)}, ylabel={Structural proof  (clauses)},
            ]

    \input{data-formatted/sub10M-clauses-subthresh}
    \input{data-formatted/sub10M-clauses-supthresh}
    \legend{
      \scriptsize \textsf{Tree ratio $\leq \; 5{.}0$},
      \scriptsize \textsf{Tree ratio $> \; 5{.}0$},
    }
    \addplot[mark=none, color=lightblue] coordinates{(100,100) (1e10,1e10)};

          \end{axis}
\end{tikzpicture}
} 
\caption{Structural (Y axis) versus monolithic (X axis) proof sizes.  The monolithic approach generated shorter proofs in most cases.}
\label{fig:clauses:mono:structural}
\end{figure}

Our first set of experiments applies full monolithic and full
structural generation to the reduced benchmark set.  Figure
~\ref{fig:seconds:mono:structural} shows a plot comparing the two
approaches.  Each axis shows the number of seconds to generate the
forward implication proof for the POG, with the X axis indicating the monolithic approach and the Y axis indicating the structural approach\@.
Data points to the left of
the diagonal line ran faster with the monolithic method, while those
to the right ran faster with the structural method.
The data are divided into
those having tree ratios below 5.0 and those having tree ratios above 5.0.  Of the 90 problems, 38 are below this tree ratio, and 52 are above.
As can be seen there is some correlation between the relative performance of the two approaches and the tree ratio.
For the 90 problems:
\begin{itemize}
\item For those with tree ratios below 5.0, 26 ran faster with monolithic generation, 11 with structural, and 1 tied.
\item For those with tree ratios above 5.0, 12 ran faster with monolithic generation, 39 with structural, and 1 tied.
\end{itemize}

Figure~\ref{fig:clauses:mono:structural} shows the comparative proof
sizes (in clauses) for the two approaches.  As can be seen, the
monolithic approach tends to generate shorter proofs.  For the 90
problems, 72 had smaller proofs with monolithic generation and 18 with
structural.  There is little correlation between
the relative proof sizes and the tree ratios.

Based on the results for the reduced benchmark set, we devised the following selection rule: when the
tree ratio for the POG is at most $5.0$, use the monolithic approach, otherwise
use the structural approach.  That would yield the better choice, in terms of runtime, for
65 of the 90 cases.
Our data set was too sparse to do more tuning, including a more refined threshold selection.

We tried a variety of hybrid approaches, where the proof generator starts at the
top using a structural approach and then switches to a monolithic
approach once the tree size for a node drops below some threshold.
This was helpful for very large problems, but setting a low threshold (tree size less than $10^6$) consistently led to poorer runtime performance.
We also found that the SAT solver could not reliably handle problems with more than $10^7$ clauses.
We therefore refined the rule for a hybrid approach that operates as follows:
\begin{enumerate}
\item With a bottom-up traversal of the graph, label each node by its tree size.
\item Compute the total size of the graph and the tree ratio of the root.
\item Proceed with proof generation with the following rules
  \begin{enumerate}
  \item If the tree ratio is at most 5.0, and the tree size of the root is below $10^6$, do the entire proof generation with a monolithic approach
  \item If the tree ratio is at most 5.0, and the tree size of the root is above $10^6$, start with a structural approach and shift to a monolithic approach once the tree size at a node is below $10^6$.
  \item If the tree ratio is above 5.0, do the entire proof generation
    with a structural approach.
  \end{enumerate}
\end{enumerate}
Unless noted otherwise, the remainder of our experimental data is based on this approach.

\subsection{Toolchain Performance Evaluation}
\label{sect:experimental:performance}

\begin{figure}
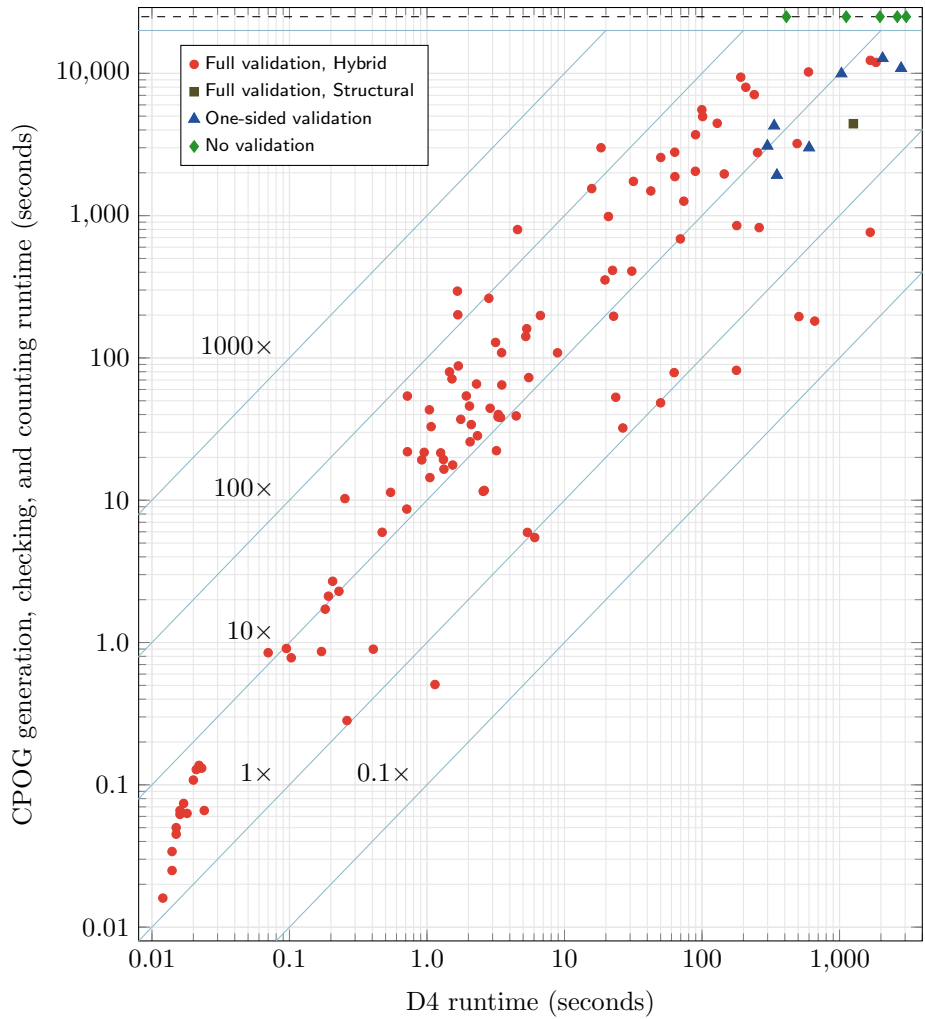

\centering{%
\begin{tikzpicture}
  \begin{axis}[mark options={scale=0.55},height=14cm,width=12cm,grid=both, grid style={black!10}, ymode=log,
      legend style={at={(0.37,0.96)}},
      legend cell align={left},
                              xmode=log,
                              xmin=0.008,xmax=4000,
                              xtick={0.01, 0.1,1.0,10,100,1000,10000}, xticklabels={0.01, 0.1, 1.0, 10, 100, {1,000}, {10,000}},
                              ymin=0.008, ymax=29000,
                              ytick={0.01, 0.1,1.0,10,100,1000,10000,100000}, yticklabels={0.01, 0.1, 1.0, 10, 100, {1,000},{10,000},{100,000}},
                              xlabel={D4 runtime (seconds)}, ylabel={CPOG generation, checking, and counting runtime (seconds)},
            ]
    \input{data-formatted/seconds-summary}
    \input{data-formatted/seconds-summary-structured-only}
    \input{data-formatted/seconds-summary-onesided}
    \input{data-formatted/seconds-summary-failures}
    \legend{
      \scriptsize \textsf{Full validation, Hybrid},
      \scriptsize \textsf{Full validation, Structural},
      \scriptsize \textsf{One-sided validation},
      \scriptsize \textsf{No validation},
    }
    \addplot[mark=none, dashed] coordinates{(0.001,25000) (4000, 25000)};) 
    \addplot[mark=none, color=lightblue] coordinates{(0.001,20000) (4000,20000)};
    \addplot[mark=none, color=lightblue] coordinates{(0.01,0.001) (10000.0,1000.0)};
    \addplot[mark=none, color=lightblue] coordinates{(0.001,0.001) (10000.0,10000.0)};
    \addplot[mark=none, color=lightblue] coordinates{(0.001,0.01) (2000,20000)};
    \addplot[mark=none, color=lightblue] coordinates{(0.001,0.10) (200, 20000)};
    \addplot[mark=none, color=lightblue] coordinates{(0.001,1.0) (20, 20000)};
    \node[left] at (axis cs: 0.9,0.12) {$0{.}1\times$};
    \node[left] at (axis cs: 0.09,0.12) {$1\times$};
    \node[left] at (axis cs: 0.09,1.2) {$10\times$};
    \node[left] at (axis cs: 0.09,12.0) {$100\times$};
    \node[left] at (axis cs: 0.09,120.0) {$1000\times$};
       \end{axis}

\end{tikzpicture}
} 
\caption{Combined runtime for CPOG proof generation, checking, and counting  as function of D4 runtime.  Timeouts are shown as points on the dashed line.
Full verification completed for 111 of the 123 benchmark problems.  The median ratio between the two times for the completed problems was $12.5$.}
\label{fig:d4:cpog}
\end{figure}

Figure~\ref{fig:d4:cpog} shows the performance of our toolchain for
the 123 problems for which \dfour{} completed within 4,000
seconds.  This figure shows the runtime for \dfour{} on the X axis and the runtime for the toolchain on the Y axis.
The toolchain included proof generation, proof checking with
the prototype checker, and counting computation.  The counting
computation included unweighted model counting for each problem, plus
weighted model counting for those from the weighted model counting
competition.  We allowed a maximum of 10,000 seconds for the
toolchain.  For those problems that failed to complete within the time
limit, we attempted other approaches.  For those with low tree ratios,
we attempted using a full structural approach.  For those where we
could not obtain a complete proof, we attempted a one-sided proof,
generating only the reverse implication proof.  The results can be summarized as follows:
\begin{itemize}
\item Of the 123 problems, 110 were completed using the hybrid approach.
\item One additional problem completed with the structural approach (as well as by using a hybrid approach with the tree size limit set to $10^5$.)
\item For seven others, we were able to generate and check a one-sided proof.
\item For five problems, no form of validation succeeded.  This included the one for which the POG was too large to encode the clause identifiers.
\end{itemize}
As one might expect, the largest problems proved to be the most
challenging.  Of the four with more than $10^9$ defining clauses, one
completed with a one-sided proof, while the other three had no form of
validation.

Figure~\ref{fig:d4:cpog} also allows comparing the time to validate
the output of the knowledge compiler relative to the time for the
compiler itself.  (The counting computations had negligible impact on
the overall toolchain performance.)  For the 111 problems for which
full proofs were generated and checked, the ratio between these two
times ranged between $0.27$ (i.e., validation was $3.64\times$ faster
than generation) and $177.0$, with a median of $12.5$.  The ones with
very high ratios tended to be ones with very few models, and so most
of the proof generation time was spent generating unsatisfiability
proofs.

It is encouraging that we could validate the results of a knowledge
compiler for all but the largest problems.  Nonetheless,
the high ratio between our toolchain time and the time required by the
compiler indicates that validation comes at a significant cost.  By
contrast, modern SAT solvers incur only a small performance penalty
when generating proofs of unsatisfiability~\cite{heule:fmcad:2013}.
With the advent of solvers that also generate hints for the proof
steps~\cite{biere:sat:2023}, the proof checking overhead has also
become very small.

\begin{figure}
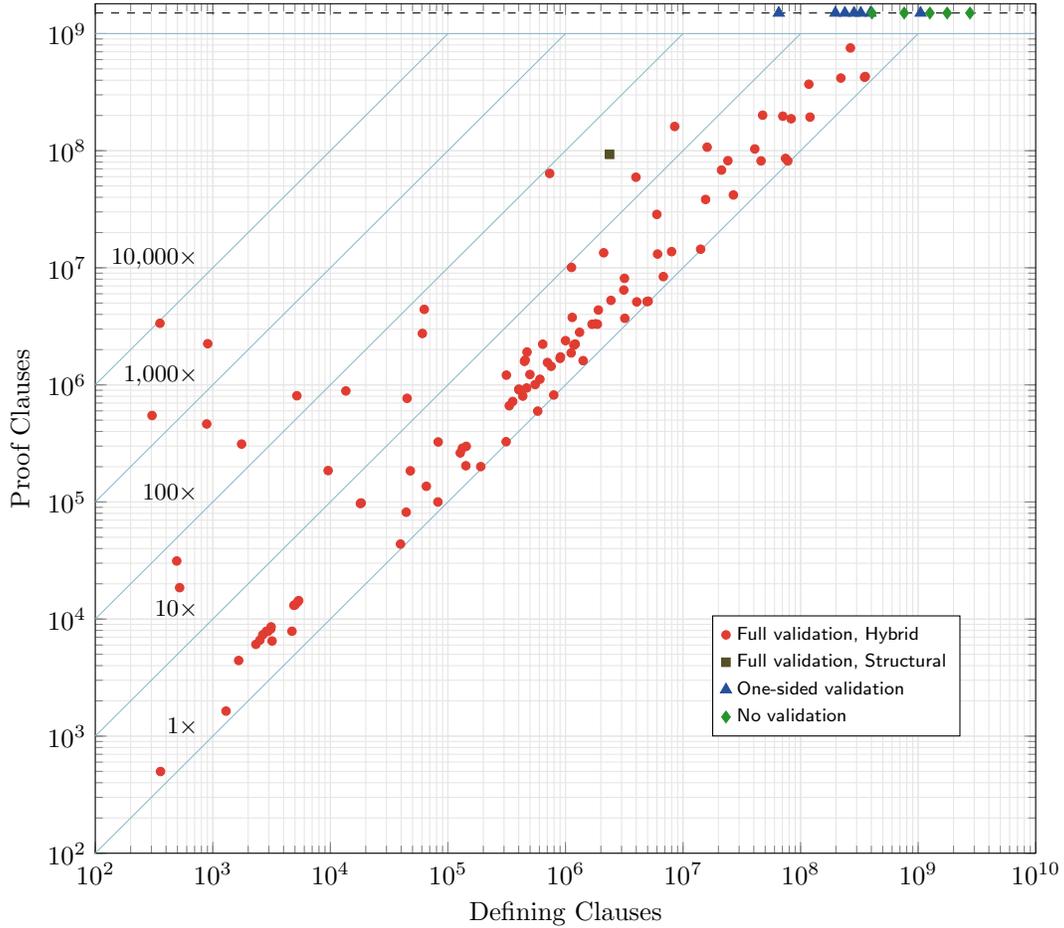

\centering{%
\begin{tikzpicture}
  \begin{axis}[mark options={scale=0.55},height=11cm,width=12cm,grid=both, grid style={black!10},
      legend style={at={(0.92,0.28)}},
      legend cell align={left},
                              x post scale=1.2, y post scale=1.2,
                              xmode=log,xmin=100,xmax=1e10, 
                              xtick={100,1000, 10000, 100000, 1000000, 10000000, 100000000, 1000000000, 1e10}, 
                              xticklabels={$10^2$,$10^3$,$10^4$,$10^5$,$10^6$,$10^7$,$10^8$,$10^9$, $10^{10}$},
                              ymode=log, ymin=100, ymax=1.8e9,
                              ytick={100,1000, 10000, 100000, 1000000, 10000000, 100000000, 1000000000},
                              yticklabels={$10^2$,$10^3$,$10^4$,$10^5$,$10^6$,$10^7$,$10^8$,$10^9$},
                              xlabel={Defining Clauses}, ylabel={Proof Clauses},
            ]
    \input{data-formatted/clauses-summary}
    \input{data-formatted/clauses-summary-structured-only}
    \input{data-formatted/clauses-summary-onesided}
    \input{data-formatted/clauses-summary-failures}

    \legend{
      \scriptsize \textsf{Full validation, Hybrid},
      \scriptsize \textsf{Full validation, Structural},
      \scriptsize \textsf{One-sided validation},
      \scriptsize \textsf{No validation},
    }
    \addplot[mark=none, color=lightblue] coordinates{(100,100) (1e9,1e9)};
    \addplot[mark=none, color=lightblue] coordinates{(100,1000) (1e8,1e9)};
    \addplot[mark=none, color=lightblue] coordinates{(100,10000) (1e7,1e9)};
    \addplot[mark=none, color=lightblue] coordinates{(100,100000) (1e6,1e9)};
    \addplot[mark=none, color=lightblue] coordinates{(100,1000000) (1e5,1e9)};
    \addplot[mark=none, color=lightblue] coordinates{(100,1e9) (1e10,1e9)};
    \addplot[color=black,dashed] coordinates{(1e2,1.5e9) (1e10,1.5e9)};
    \node[left] at (axis cs: 900,1200) {\small $1\times$};
    \node[left] at (axis cs: 900,12000) {\small $10\times$};
    \node[left] at (axis cs: 900,120000) {\small $100\times$};
    \node[left] at (axis cs: 900,1200000) {\small $1{,}000\times$};
    \node[left] at (axis cs: 900,12000000) {\small $10{,}000\times$};

          \end{axis}
\end{tikzpicture}
} 
\caption{Total number of clauses in CPOG file as function of number of defining clauses.  The median ratio of the two was $2.29$.}
\label{fig:defining:total}
\end{figure}

Figure~\ref{fig:defining:total} compares the total number of clauses
in the CPOG representation (Y axis) versus the number of defining
clauses (X axis).  Since the former include the latter, the ratio
between these cannot be less than $1.0$.  The ratios ranged between
$1.02\times$ and $9460.2\times$.  Again, the largest ratios were for
problems with very few models, and hence most of the steps were for
the unsatisfiability proofs in the literal justifications.  The median
ratio was $2.29\times$.  This is a relatively modest overhead,
although it requires transforming the large dec-DNNF files into  even larger CPOG files.

\subsection{Toolchain Robustness Evaluation}
\label{sect:experimental:robustness}

Although the CPOG framework is very general and makes no assumptions
about how the the POG relates to the input CNF formula, our proof
generator is less general.  It requires that the POG arise from a
dec-DNNF graph.  Moreover, our structural approach requires that the
CNF formula decompose according to the dec-DNNF structure.  That is,
as it recurses downward, the simplified clauses must be encoded by the
POG subgraphs.

Our monolithic approach, on the other hand, makes no assumption about
the relation between the POG and the CNF formula.  As long as every
satisfying assignment to the CNF would, when extended, cause the POG
root to evaluate to 1, the monolithic approach can, in principle, generate a forward
implication proof.  Our reverse implication proof generation is also
independent of any structural relations between the two
representations.

We tested this hypothesis by using the preprocessing capabilities of
\dfour{} to transform the input formula into a different, but
logically equivalent clausal representation.  \dfour{} can optionally
perform three different forms of
preprocessing~\cite{lagniez:aaai:2014}.  These are designed to make
knowledge compilation more efficient, but they also have the effect of
creating a mismatch between the structure of the generated dec-DNNF
graph and the original input formula.

We used the 90 problems from the reduced benchmark set as test cases,
running \dfour{} by preprocessing with all three methods enabled
(these are referred to as ``backbone,'' ``vivification,'' and
``occElimination'') followed by knowledge compilation.  None of the
resulting POGs could be verified using the structural approach.
Setting an overall time limit of 1000 seconds for the combination of
\dfour{} (including preprocessing), proof generation, proof checking,
and counting, and using monolithic proof generation, we obtained the following results:
\begin{itemize}
\item For 7 problems, neither approach completed within 1000 seconds.
\item For 2 problems, running with preprocessing completed within the time limit, while running without did not.
\item For 1 problem, running without preprocessing completed within the time limit, while running with did not.
\item For 45 problems, both completed, with the preprocessing version running faster.
\item for 35 problems, both completed, with the preprocessing version running slower.
\end{itemize}

These results indicate that the preprocessing is only marginally
effective.  Importantly, however, they demonstrate
that our toolchain can establish the end-to-end
correctness of preprocessing plus knowledge compilation.

Even with monolithic mode, our proof generator still requires that the output of the knowledge compiler be
a dec-DNNF graph.  We discuss how it could be
generalized even further in Section~\ref{sect:extend:pog}.

\subsection{Effect of Optimizations}
\label{sect:experimental:optimize}
\begin{figure}
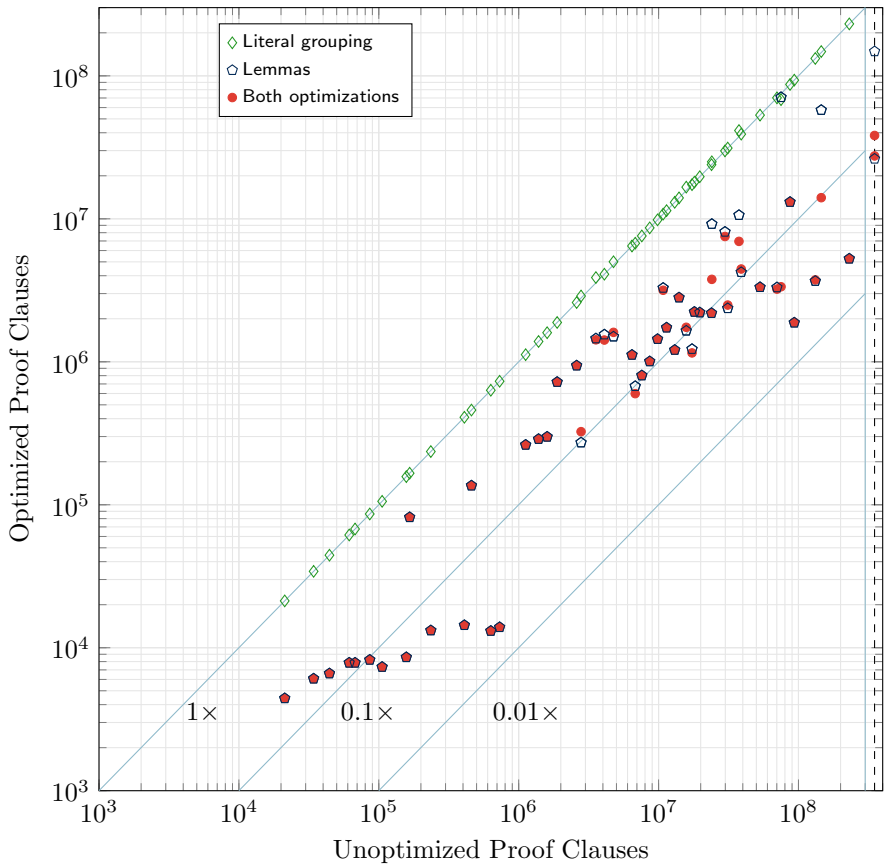


\begin{center}
\begin{tikzpicture}
  \begin{axis}[mark options={scale=0.55},grid=both, height=12cm,width=12cm, grid style={black!10}, 
      legend style={at={(0.4,0.98)}},
      legend cell align={left},
                              xmode=log,xmin=1e3,xmax=4e8, 
                              xtick={1000, 10000, 100000, 1000000, 10000000, 1e8, 1e9}, xticklabels={$10^3$,$10^4$,$10^5$,$10^6$,$10^7$, $10^8$, $10^9$},
                              xlabel={Unoptimized Proof Clauses},
                              ymode=log, ymin=1e3, ymax=3e8, 
                              ytick={1000, 10000, 100000, 1000000, 10000000, 1e8, 1e9}, yticklabels={$10^3$,$10^4$,$10^5$,$10^6$,$10^7$,$10^8$, $10^9$},
                              ylabel={Optimized Proof Clauses},
            ]
    \input{data-formatted/clauses-unoptimized-group}
    \input{data-formatted/clauses-unoptimized-lemma}
    \input{data-formatted/clauses-unoptimized-optimized}    
    \legend{
      \scriptsize \textsf{Literal grouping},
      \scriptsize \textsf{Lemmas},
      \scriptsize \textsf{Both optimizations},
    }
    \input{data-formatted/clauses-failures-unoptimized-optimized}    
    \input{data-formatted/clauses-failures-unoptimized-lemma}
    \input{data-formatted/clauses-unoptimized-lemma}
    \addplot[mark=none, color=lightblue] coordinates{(1000,1000) (3e8,3e8)};
    \addplot[mark=none, color=lightblue] coordinates{(10000,1000) (3e8,3e7)};
    \addplot[mark=none, color=lightblue] coordinates{(100000,1000) (3e8,3e6)};

    \node[right] at (axis cs: 3.6e3,3.5e3) {$1\times$};
    \node[right] at (axis cs: 4.6e4,3.5e3) {$0.1\times$};
    \node[right] at (axis cs: 5.6e5,3.5e3) {$0.01\times$};

    \addplot[mark=none, color=lightblue] coordinates{(3e8,1e3) (3e8,3e8)};
    \addplot[color=black,dashed] coordinates{(3.5e8,1e3) (3.5e8,3e8)};

          \end{axis}
\end{tikzpicture}
\end{center}
\caption{Proof clauses when one or both optimizations is enabled, versus without optimization.  Lemmas provide substantial benefit, while the results for literal grouping are mixed.}
\label{fig:optimized:lessoptimized}
\end{figure}

Section~\ref{sect:optimization} describes two optimizations for proof
generation: literal grouping and lemmas.  These optimizations are only
applied when using a structural approach, and so we focus our evaluation
on the 52 problems having tree ratios greater than 5.0 from the reduced benchmark set of 90 problems.

Figure~\ref{fig:optimized:lessoptimized} summarizes the sizes of the
CPOG representations generated for these problems with and without the
optimizations.  The X axis shows the size (in clauses) for the proof
when neither optimization is enabled, while the Y axis shows the sizes
with either one or both enabled.  The extent to which a point lies
below the diagonal line labeled ``$1\times$'' therefore indicates the
benefit of the optimizations.  Two benchmarks required lemmas to
complete.  These are indicated along the far edge of the X axis.  In
the remaining, we consider mostly the 50 benchmarks for which all four variants completed.

Literal grouping alone (the hollow diamonds clustered along the
diagonal line), has only minimal benefit.  Compared to the unoptimized
proof sizes, literal grouping yielded proofs that ranged between being
$1.10\times$ larger and $1.10\times$ smaller, with a median ratio of
$1.0$.  Although literal grouping reduces
the number of unsatisfiability proofs that must be generated, the
resulting proofs are enough larger to offset this advantage.

Using lemmas alone (the hollow pentagons), on the other hand, shows
significant benefit.  The resulting proofs were between $1.06\times$
and $52.54\times$ smaller, with a median of $7.95\times$.  In addition,
lemmas enable two benchmarks to complete that otherwise fail.
These
problems have high degrees of subgraph sharing, and so the ability to
avoid expanding the proofs into tree structures was important.

Combining literal grouping with lemmas (the solid dots)
showed a modest improvement over using lemmas alone.
Many of the solid dots coincide with or are very
close to the hollow pentagons, with some being slightly better and others begin slightly worse.
Significantly, however, several problems showed major benefit from combining the two
optimizations.  In the most extreme case, one problem had between 68
and 75 million proof clauses with either no or a single optimization, but just 3.3 million with both optimizations.
Compared to the unoptimized proofs, the combination
yielded proofs ranging from $2.03\times$ to $52.54\times$
smaller, with a median of $8.59\times$.

The runtime improvement with the optimizations was smaller than the size improvement, but still significant.
Generating shorter proofs enables the checker to run faster, and so
there is some benefit in spending more time in proof generation to
reduce the proof size.  We therefore consider the combined time to
generate and to check the proofs.  Literal grouping, on its own,
caused the toolchain to run with a range from $4.02\times$ slower to
$1.24\times$ faster, with a median slowdown of $1.70\times$, compared
to no optimization.  Lemmas, on their own, yielded speedups ranging
from $1.03\times$ to $18.04\times$, with a median of $2.78\times$.
Combining the two yielded peformances ranging from a slowdown of
$1.23\times$ to a speedup of $23.01\times$, with a median speedup of
$3.02\times$.

Overall, these results indicate lemmas provide an important optimization,
while literal grouping provides a modest benefit.

\subsection{Performance of the Formally Verified Proof Checker}
\label{sect:experimental:verified-checker}

\begin{figure}
\centering{%
\begin{tikzpicture}
  \begin{axis}[mark options={scale=0.55},height=12cm,width=12cm,grid=both, grid style={black!10},
      legend style={at={(0.22,0.97)}},
      legend cell align={left},
                              xmode=log, xmin=0.0008, xmax=10000,
                              xtick={0.001, 0.01, 0.1,1.0,10,100,1000,10000}, xticklabels={0.001, 0.01, 0.1, 1.0, 10, 100, {1,000}, {10,000}},
                              ymode=log, ymin=0.0008, ymax=10000,
                              ytick={0.001, 0.01, 0.1,1.0,10,100,1000,10000}, yticklabels={0.001, 0.01, 0.1, 1.0, 10, 100, {1,000}, {10,000}},
                              xlabel={Prototype checker (seconds)}, ylabel={Verified checker (seconds)},
            ]
    \input{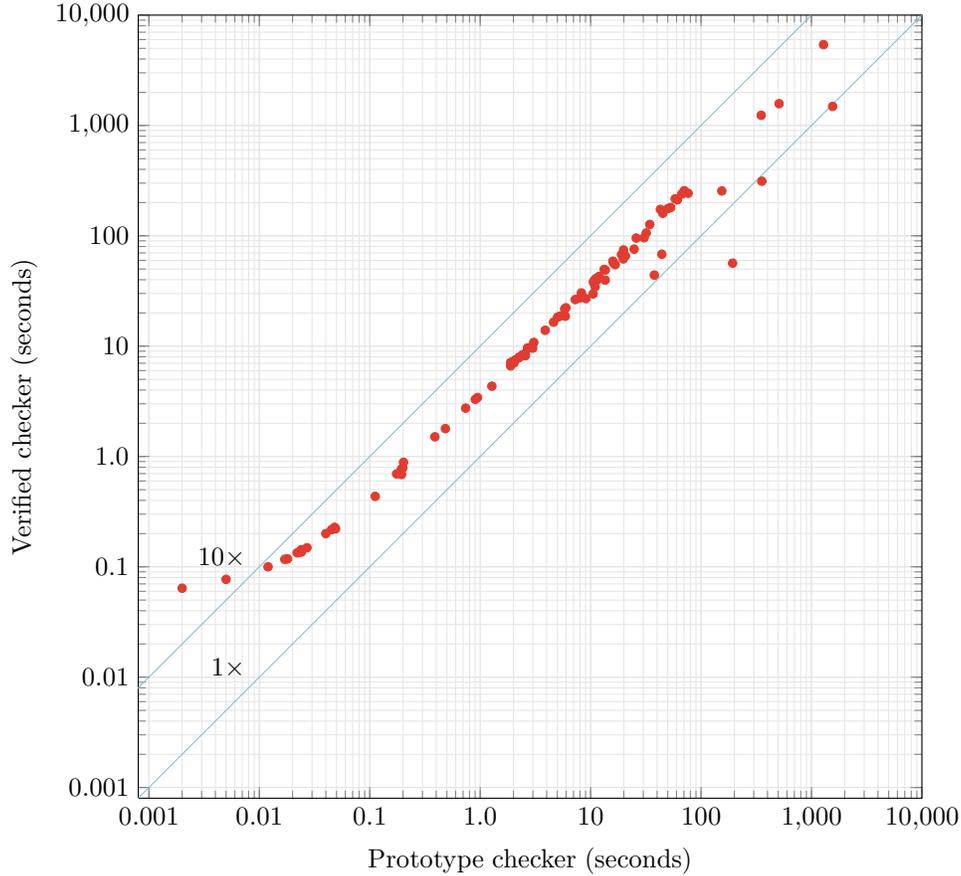}
    \addplot[mark=none, color=lightblue] coordinates{(0.0001,0.0001) (10000,10000)};
    \addplot[mark=none, color=lightblue] coordinates{(0.0001,0.001) (1000,10000)};
    \node[left] at (axis cs: 0.009,0.012) {$1\times$};
    \node[left] at (axis cs: 0.009,0.12) {$10\times$};

          \end{axis}
\end{tikzpicture}
} 
\caption{Times for Verified Checker versus Prototype Checker.  Both show similar scaling.}
\label{fig:lean:compare}
\end{figure}

Our prototype proof checker is fairly simple and has shown itself to
be reliable, but we have not subjected it to rigorous, adversarial
testing.  Using our verified checker removes any  doubt about
the trustworthiness of the compiled result.  For the 90 problems
from the reduced set, we generated CPOG files using the hybrid
approach and ran both checkers.  Figure~\ref{fig:lean:compare}
summarizes the results, with the runtime for the prototype checker on
the X axis and for the verified checker on the Y axis.

We can see in this figure that the verified checker has a startup time
of around 70 milliseconds, causing it to run much slower compared to
the prototype checker on the very small problems.  If we consider only the 76 problems requiring
more than 0.1 seconds with the prototype checker, we see that the
verified checker runs between $3.42\times$ faster and $4.39\times$ slower
than the prototype, with a median of $3.54\times$ slower.

Significantly, the relative performance
remains constant even for the larger proofs, showing that the two
programs have similar scaling properties.

\subsection{Comparison to Other Validation Frameworks}
\label{sect:experimental:comparisons}

As described in Section~\ref{sect:related}, two other verification
frameworks have been developed that are relevant to ours: the
\cdfour{} framework~\cite{capelli:sat:2019,capelli:aaai:2021},
designed for the \dfour{} knowledge compiler, and the \mice{}
framework~\cite{fichte:sat:2022}, designed to verify unweighted model
counters.  Here we compare how they perform on all 123 problems in the full benchmark set.

\begin{figure}
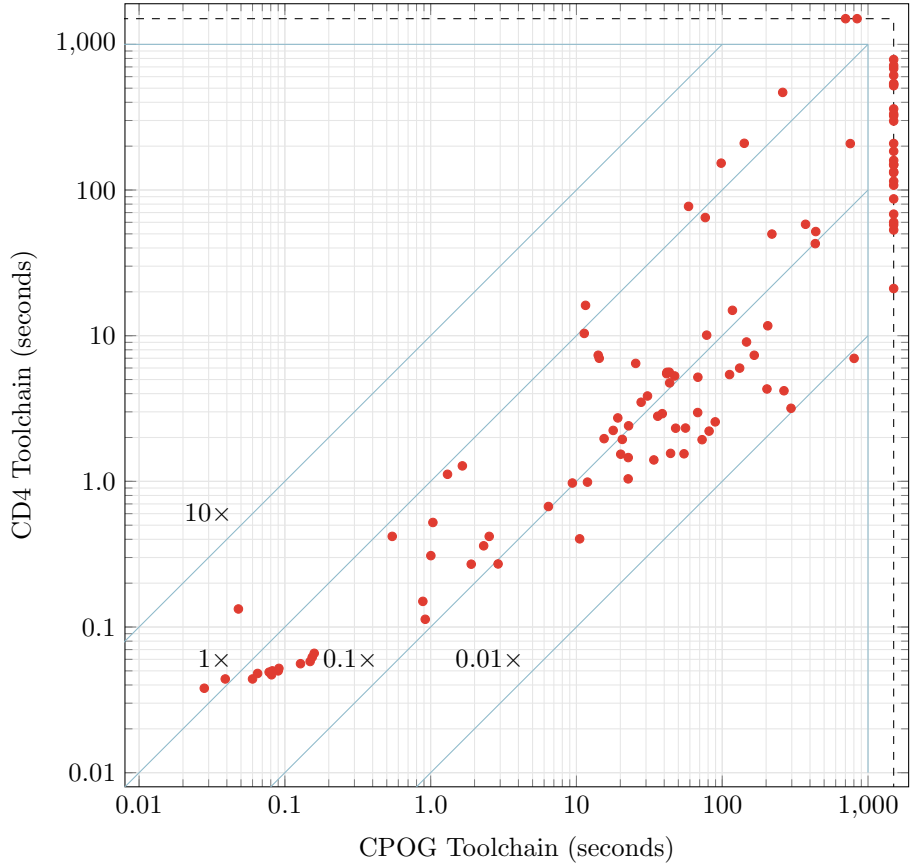

\centering{%
\begin{tikzpicture}
  \begin{axis}[mark options={scale=0.55},height=12cm,width=12cm,grid=both, grid style={black!10}, ymode=log,
      legend style={at={(0.22,0.98)}},
      legend cell align={left},
                              xmode=log,xmin=0.008,xmax=1900,
                              xtick={0.01, 0.1,1.0,10,100,1000}, xticklabels={0.01, 0.1, 1.0, 10, 100,{1,000}},
                              ymin=0.008, ymax=1900,
                              ytick={0.01, 0.1,1.0,10,100,1000}, yticklabels={0.01, 0.1, 1.0, 10, 100, {1,000}},
                              xlabel={CPOG Toolchain (seconds)}, ylabel={CD4 Toolchain (seconds)},
            ]
    \input{data-formatted/time-toolchain-compare-cd4-timeout}
    \input{data-formatted/time-toolchain-compare-pog-fail}
    \input{data-formatted/time-toolchain-compare}
    \addplot[mark=none, color=lightblue] coordinates{(0.001,0.00001) (1000.0,10.0)};
    \addplot[mark=none, color=lightblue] coordinates{(0.001,0.0001) (1000.0,100.0)};
    \addplot[mark=none, color=lightblue] coordinates{(0.001,0.001) (1000.0,1000.0)};
    \addplot[mark=none, color=lightblue] coordinates{(0.001,0.01) (100, 1000)};
    \addplot[color=black,dashed] coordinates{(0.001,1500) (1500,1500)};
    \addplot[color=black,dashed] coordinates{(1500,0.001) (1500,1500)};
    \addplot[mark=none, color=lightblue] coordinates{(0.001,1000) (1000,1000)};
    \addplot[mark=none, color=lightblue] coordinates{(1000,0.01) (1000,1000)};
    \node[left] at (axis cs: 5.0,0.06) {$0{.}01\times$};
    \node[left] at (axis cs: 0.5,0.06) {$0{.}1\times$};
    \node[left] at (axis cs: 0.05,0.06) {$1\times$};
    \node[left] at (axis cs: 0.05,0.6) {$10\times$};

          \end{axis}
\end{tikzpicture}
} 
\caption{Times for CD4 Toolchain versus CPOG Toolchain.  Times include knowledge compilation, proof generation, and checking.  CD4 generally scales better.}
\label{fig:cd4:cpog}
\end{figure}

Running \cdfour{} involves running \dfour{} with appropriate
arguments.\footnote{This is possible with the original version of \dfour{}, available 
at \url{https://github.com/crillab/d4}.  It was not incorporated into the more recent version, available at
\url{https://github.com/crillab/d4v2}.}  Checking the results requires running two checkers: one
for the annotated dec-DNNF graph, plus \dtrim{} for the generated proof
clauses.  The first checker is not available in any public repository.  We used a copy supplied to us by the authors.
The combined toolchain therefore involves running the
knowledge compiler and the two proof checkers.  For comparison, we
consider the time for our complete toolchain, including running
\dfour{}, the proof generator, and the prototype proof checker.  For
both toolchains, we set a time limit of 1,000 seconds.  We ran both
toolchains for all 123 problems.

Figure~\ref{fig:cd4:cpog} compares times for the toolchains, with
those for our toolchain on the X axis and those for the \cdfour{}
toolchain on the Y axis.
The results can be summarized as follows:
\begin{itemize}
\item Both toolchains completed for 82 problems, with 8 running faster with our toolchain and 74 running faster with the \cdfour{} toolchain.
  Overall, our toolchain ranged from $2.77\times$ faster to $114.91\times$ slower, with a median of running $7.81\times$ slower.
\item Our toolchain completed 2 problems for which the \cdfour{} toolchain did not complete within 1000 seconds.
\item The \cdfour{} toolchain completed 26 problems for which our toolchain did not complete within 1000 seconds.
\item Neither toolchain completed for 13 problems.
\end{itemize}
Clearly, \cdfour{} has better overall scaling and
performance.  Even with a time limit of 1000 seconds, it was able to handle all but 15 of our 123 problems.

The \cdfour{} toolchain has impressive performance, but as a general tool it has significant shortcomings.
It relies strongly on the inner
workings of the knowledge compiler.  It cannot even verify its own
output when preprocessing is enabled.  Furthermore, even having corrected
the known flaw, there is no guarantee that their framework is sound or
that their checker is correct.

\begin{figure}
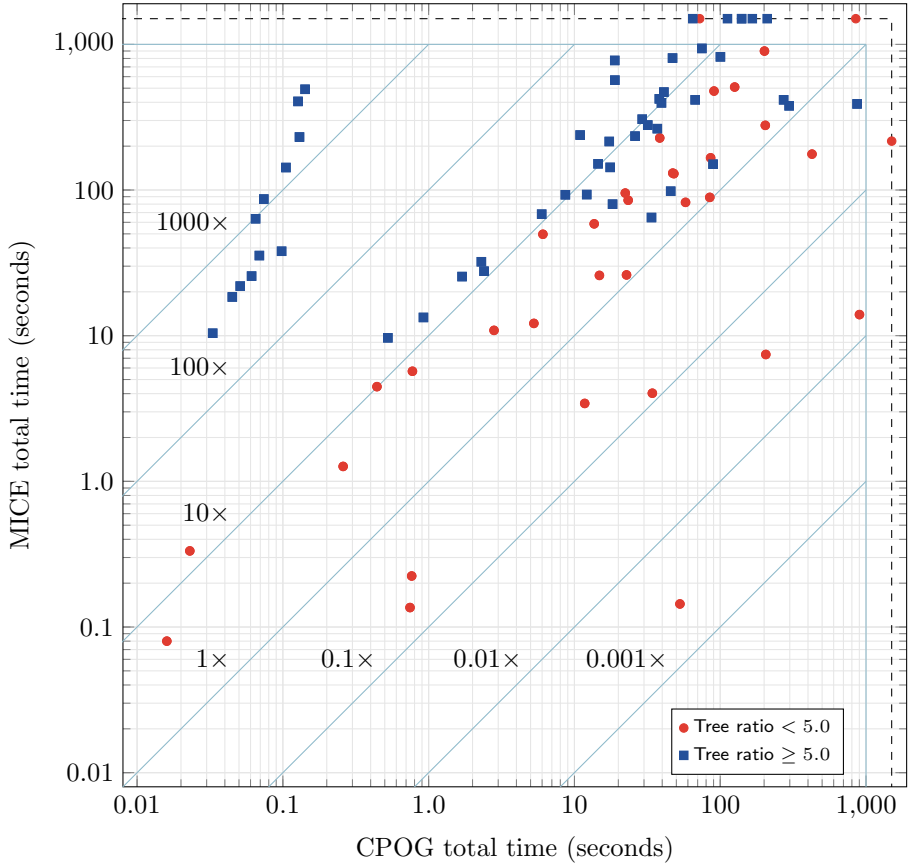

\centering{%
\begin{tikzpicture}
  \begin{axis}[mark options={scale=0.55},height=12cm,width=12cm,grid=both, grid style={black!10},
      legend style={at={(0.92,0.10)}},
      legend cell align={left},
                              xmode=log,xmin=0.008,xmax=1900,
                              xtick={0.01, 0.1,1.0,10,100,1000}, xticklabels={0.01, 0.1, 1.0, 10, 100,{1,000}},
                              ymode=log, ymin=0.008, ymax=1900,
                              ytick={0.01, 0.1,1.0,10,100,1000}, yticklabels={0.01, 0.1, 1.0, 10, 100, {1,000}},
                              xlabel={CPOG total time (seconds)}, ylabel={MICE total time (seconds)},
            ]
    \input{data-formatted/time-mice-compare-subthresh}
    \input{data-formatted/time-mice-compare-supthresh}
    \legend{
      \scriptsize \textsf{Tree ratio $< 5.0$},
      \scriptsize \textsf{Tree ratio $\geq 5.0$},
    }
    \input{data-formatted/time-mice-compare-subthresh}
    \input{data-formatted/time-mice-timeout-subthresh}
    \input{data-formatted/time-pog-timeout-subthresh}
    \input{data-formatted/time-mice-compare-supthresh}
    \input{data-formatted/time-mice-timeout-supthresh}

    \addplot[color=black,dashed] coordinates{(0.001,1500) (1500,1500)};
    \addplot[color=black,dashed] coordinates{(1500,0.001) (1500,1500)};
    \addplot[mark=none, color=lightblue] coordinates{(0.001,1000) (1000,1000)};
    \addplot[mark=none, color=lightblue] coordinates{(1000,0.001) (1000,1000)};

    \addplot[mark=none, color=lightblue] coordinates{(1.0,0.001) (1000.0,1.0)};
    \addplot[mark=none, color=lightblue] coordinates{(0.1,0.001) (1000.0,10.0)};
    \addplot[mark=none, color=lightblue] coordinates{(0.01,0.001) (1000.0,100.0)};
    \addplot[mark=none, color=lightblue] coordinates{(0.001,0.001) (1000.0,1000.0)};
    \addplot[mark=none, color=lightblue] coordinates{(0.001,0.01) (100,1000)};
    \addplot[mark=none, color=lightblue] coordinates{(0.001,0.1) (10,1000)};
    \addplot[mark=none, color=lightblue] coordinates{(0.001,1.0) (1,1000)};
    \node[left] at (axis cs: 50,0.06) {$0{.}001\times$};
    \node[left] at (axis cs: 5.0,0.06) {$0{.}01\times$};
    \node[left] at (axis cs: 0.5,0.06) {$0{.}1\times$};
    \node[left] at (axis cs: 0.05,0.06) {$1\times$};
    \node[left] at (axis cs: 0.05,0.6) {$10\times$};
    \node[left] at (axis cs: 0.05,6.0) {$100\times$};
    \node[left] at (axis cs: 0.05,60.0) {$1000\times$};

          \end{axis}
\end{tikzpicture}
} 
\caption{Running Time for MICE versus our proof chains.  Times include proof generation, checking, and counting.  Timeouts are shown as points on the dashed lines.
MICE is especially weak on problems with high tree ratios.}
\label{fig:mice}
\end{figure}

Running \mice{} on the output of a knowledge compiler requires running
two programs: \progname{nnf2trace}, a proof generator for dec-DNNF
graphs, and \progname{sharptrace}, a checker for the generated proofs.\footnote{Both programs were downloaded from \url{https://github.com/vroland}.}

The results for the reduced set of 90 problems is shown in
Figure~\ref{fig:mice}, comparing the time to generate and check
the proofs with our framework on the X axis, and the time to do so
with the \mice{} tools on the Y axis.  Both were set to have a time limit
of 1000 seconds.  
The results can be summarized as follows:
\begin{itemize}
\item Both toolchains completed for 75 problems, with 66 running faster with our toolchain and 9 running faster with the \mice{} toolchain.
  Overall, our toolchain ranged from $3461\times$ faster to $368\times$ slower, with a median of running $7.67\times$ faster.
\item Our toolchain completed 7 problems for which the \mice{} toolchain did not complete within 1000 seconds.
\item The \mice{} toolchain completed 1 problem for which our toolchain did not complete within 1000 seconds.
\item Neither toolchain completed for 7 problems.
\end{itemize}

One shortcoming of the \mice{} framework is highlighted by the
division of the data points in Figure~\ref{fig:mice} according to
tree ratios.  Those with tree ratios above 5.0 consistently performed
poorly for \mice{}, with 5 exceeding the time limit and 43 requiring
more time than with our toolchain.  Only 1 problem above this
threshold ran faster with \mice{} than with ours.  These are the
problems with significant amounts of sharing in the subgraph.  Our
toolchain exploits this sharing by generating and using lemmas for the
shared subgraphs.  \Mice{}, on the other hand, has no mechanism for
reusing results, effectively expanding the graphs into trees.  

Overall, these results indicate that the \mice{} framework has serious
performance limitations, due in part to its inability to efficiently
exploit the sharing of subgraphs.  In addition, the \mice{} proof generator relies
strongly on the means by which the knowledge compiler output was
generated.  For example, it cannot perform an end-to-end verification
of the combination of preprocessing and knowledge compilation.  Other shortcomings include
the lack of formal verification for the framework or the checker, and that the framework can only
validate the unweighted model count.

\section{Extensions}
\label{sect:extensions}

We are hopeful that having checkable proofs for knowledge compilers
will allow them to be used in applications where high levels of trust
are required, and that it will provide a useful tool for developers of
knowledge compilers.  Our current implementation only handles the
outputs of the \dfour{} knowledge compiler, and it
supports only queries that can be computed via ring evaluation.  Here
we discuss ways to extend both capabilities.

\subsection{Validating Arbitrary POGs}
\label{sect:extend:pog}

Extending our proof generator to other
knowledge compilers that generate \decdnnf{}, such as
\progname{Dsharp}~\cite{muise:cai:2012}, requires simply extending
the parser.  Some knowledge compilers, however, generate
representations that cannot be directly encoded into \decdnnf{}.  For
example, the Sentential Decision Diagram representation introduced by
Darwiche~\cite{darwiche:ijcai:2011} can readily be translated into \detdnnf{},
but with the possibility that some sum nodes will not have 
associated decision variables.

Extending our tool to handle arbitrary POGs, including \detdnnf{} as a
subset, could be done with modest effort.  Our monolithic approach can
generate forward implication proofs for this more general form.  Our
method for generating reverse implication proofs currently handles
\detdnnf{} formulas~\cite{darwiche:jair:2002,capelli:sat:2019}, but not formulas with negations.
Extending it to POGs would require marking nodes for both negative and positive polarities.
The proof generator must also generate mutual
exclusion generate proofs for each sum node declaration.  This could be done
with a proof-generating SAT solver.  That is, for child nodes
$\nodeu_0$ and $\nodeu_1$, it would generate a CNF formula $\theta_c$
consisting of the defining clauses for the subgraphs having $\nodeu_0$
and $\nodeu_1$ as roots, and run a SAT solver on
$\simplify{\theta_c}{\{u_0, u_1\}}$, the formula that would be
satisfied by an extended assignment $\eassign$ that assigns value $1$
to both children.  The proof of unsatisfiability can then be
translated into a series of clause additions, adding literals
$\obar{u}_0$ and $\obar{u}_1$ to each proof clause.  The hint for the
final proof step then serves as the hint for the mutual exclusion
proof in the sum declaration.

\subsection{Generalizing to Semirings}
\label{sect:extend:semiring}

The formulation of algebraic model
counting by Kimmig, et al.~\cite{kimmig:jal:2017} is more general than ours.  It allows
the algebraic structure to be a \emph{semiring}.  A commutative
semiring $\semiring$ obeys all properties of a commutative ring, except that the
elements of the set need not have additive inverses.
We can define the {\em semiring evaluation problem} as computing
  \begin{equation}
    \begin{array}{rcl}
    \srep(\phi, w) & = & \sum_{\alpha \in \modelset(\phi)} \;\; \prod_{\lit \in \alpha} w(\ell) \label{eqn:srep}
    \end{array}
  \end{equation}
  where sum $\sum$ is computed according to the semiring addition operation $\radd$ and product $\prod$
  is computed according to the semiring product operation $\rmul$.

  As an example, consider the formulation of the weighted model counting
  computation in Section~\ref{sect:ring}, but using $\max$ as the sum
  operation, rather than addition.  The computation would then yield
  the maximum weight for all satisfying assignments, rather than their
  sum.

Semiring evaluation can be performed via knowledge compilation
by requiring 
that the representation generated by
the compiler be in negation normal form, and that it obey a property known as
\emph{smoothness}~\cite{darwiche:jair:2002,shih:nips:2019}.  Within our formulation,
a partitioned-operation formula is smooth when all arguments to each
sum operation have identical dependency sets.  That is, every sum
operation $\bigvee_{1 \leq i \leq k} \phi_i$ has
$\dependencyset(\phi_i) = \dependencyset(\phi_1)$ for $1 < i \leq k$.  Smoothness can be ensured by adding redundant formulas to
artificially introduce variables.  For example, if subformula $\phi_i$
lacks having variable $x$ in its dependency set, it can be replaced by
$(x \por \obar{x}) \pand \phi_i$.  When a knowledge compiler generates
a representation in negation normal form that is smooth, then a semiring evaluation of the formula
can proceed by first assigning each literal $\lit$ the value $w(\lit)$.
Then the product and sum operations are evaluated in manners analogous to
(\ref{eqn:ring:product}) and
(\ref{eqn:ring:sum}).

Our POG representation can support evaluation of semiring formulas by
imposing the restriction that the POG is in negation normal form and that it is smooth.
Given a smoothed
\decdnnf{} graph generated by a knowledge compiler, our toolchain will
convert this into a smooth POG in negation normal form and verify its
equivalence to the input formula.  Full verification would also require checking that the POG is smooth.  We must also
extend the formal derivation to ensure soundness and to create a formally verified checker.

\section{Concluding Remarks}
\label{sect:future}

This paper demonstrates a method for certifying the equivalence of two
different representations of a Boolean formula: an input formula
represented in conjunctive normal form, and a compiled representation
that can then be used to extract useful information about the formula,
including its weighted and unweighted model counts.  It builds on the
extensive techniques that have been developed for clausal
proof systems, including extended resolution and reverse unit propagation, as well as established tools, such as
proof-generating SAT solvers.

Our experiments demonstrate that our toolchain can already handle
problems nearly at the limits of current knowledge compilers.  Further
engineering and optimization of our proof generator and checker could
improve their performance and capacity substantially.  We also show
that, by using monolithic proof generation, our toolchain can be
agnostic to the means by which the knowledge compiler created a
\decdnnf{} representation of the input formula.  This generality, plus
the fact that our toolchain has been formally verified, provides a
major improvement over previous methods for checking the outputs of
knowledge compilers and model counters.

\section*{Acknowledgments}

Funding for Randal E. Bryant and Marijn J. H. Heule was provided by the National Science Foundation, NSF grant CCF-2108521.
Funding for Wojciech Nawrocki and Jeremy Avigad was provided by the Hoskinson Center for Formal Mathematics at Carnegie Mellon University.

\bibliography{references}

\end{document}